\documentclass[11pt]{article}

\usepackage{natbib}
\usepackage{graphicx, verbatim}
\usepackage{amssymb}
\usepackage{alltt}
\usepackage{amsmath, amssymb, amsfonts, amscd, xspace, pifont, amsthm}
\usepackage{color}

\usepackage{braket}
\usepackage{mathrsfs}
\usepackage[ruled,linesnumbered]{algorithm2e}
\newtheorem{theorem}{Theorem}
\newtheorem{lemma}{Lemma}

\usepackage{authblk}

\def\threeImages#1#2#3#4#5#6#7#8#9 
{
  \centerline{\hfill\makebox[#2]{#3}\hfill\makebox[#5]{#6}\hfill\makebox[#8]{#9}\hfill}
  \centerline{\hfill
    \includegraphics[width=#2]{#1}
    \hfill
    \includegraphics[width=#5]{#4}
    \hfill
    \includegraphics[width=#8]{#7}
    \hfill}
}

\def\twoImages#1#2#3#4#5#6 
{
  \centerline{\hfill\makebox[#2]{#3}\hfill\makebox[#5]{#6}\hfill}
  \centerline{\hfill
    \includegraphics[width=#2]{#1}
    \hfill
    \includegraphics[width=#5]{#4}
    \hfill}
}

\newenvironment{remark}[1][Remark]{\begin{trivlist}
\item[\hskip \labelsep {\bfseries #1}]}{\end{trivlist}}

\newcommand{\tb}{\textbf}

\begin{document}
\title{Adaptive  Algorithm for Quantum Amplitude Estimation}
\author[1]{Yunpeng Zhao}
\author[1]{Haiyan Wang}
\author[1]{Kuai Xu}
\author[1]{Yue Wang}
\author[2]{Ji Zhu}
\author[1]{Feng Wang}
\affil[1]{School of Mathematical and Natural Sciences, Arizona State University}
\affil[2]{Department of Statistics, University of Michigan}

\renewcommand\Authands{ and }
\date{}
\maketitle
\begin{abstract}
	Quantum amplitude estimation is a key sub-routine of a number of quantum algorithms with various applications. We propose an adaptive algorithm for interval estimation of amplitudes.  The quantum part of the algorithm is  based only on Grover's algorithm.  The key ingredient is the introduction of an adjustment factor, which adjusts the amplitude of good states such that the  amplitude after the adjustment, and  the original amplitude, can be estimated without ambiguity in the subsequent step. We show with numerical studies that the proposed algorithm uses a similar number of quantum queries to achieve the same level of precision $\epsilon$ compared to state-of-the-art algorithms, but the classical part, i.e., the non-quantum part, has   substantially lower  computational complexity.  We rigorously prove that the number of oracle queries achieves $O(1/\epsilon)$, i.e., a quadratic speedup over  classical Monte Carlo sampling, and the computational complexity of the classical part achieves $O(\log(1/\epsilon))$, both up to a  double-logarithmic factor.

\end{abstract}

\section{Introduction}
Quantum computers have the potential to perform high-speed computations based on a fundamentally different manner of storing and processing data -- quantum superpositions  and unitary transformations. The reader is referred to \cite{nielsen} for a comprehensive introduction to quantum computing and \cite{wang2022quantum,Wang2022When} for quantum computing in a context of statistics and data science. A major milestone in quantum computing is  the discovery of a polynomial-time quantum algorithm for integer factorization \citep{shor1994algorithms}, which is almost exponentially faster than the most efficient known classical algorithm \citep{pomerance1996tale}. Another famous quantum algorithm is Grover's algorithm \citep{grover1996fast}, which  finds with high probability the unique input to a black box function defined on $\{0,...,N-1\}$ that gives a particular output, using $O(\sqrt{N})$ queries.  Although only achieving a quadratic speedup over a classical brute-force search, Grover's algorithm makes no  assumption on the function other than the number of the solutions (later relaxed by \cite{brassard2002quantum}), and therefore has a wide range of potential applications \citep{ambainis2004quantum,sun2014quantum,zhong2021best}. In recent years, quantum algorithms have been developed for various domains including finance \citep{hong2014monte,herman2022survey}, chemistry \citep{cao2019quantum},  optimization \citep{durr1996quantum,kochenberger2014unconstrained,wang2016quantum,hu2020quantum}, machine learning \citep{ramezani2020machine}, and high-dimensional statistics \citep{zhong2021best}, among others.

In this paper, we focus on the amplitude estimation problem introduced by \cite{brassard2002quantum}. Suppose  the basis states in a finite-dimensional complex Hilbert space\footnote{Hilbert space typically refers to an infinite-dimensional function space in mathematics. In quantum computing, finite-dimensional complex Hilbert spaces are usually considered, which are simply finite-dimensional complex inner product spaces.} are partitioned into two sets, called \textit{good} states and \textit{bad} states. Given a quantum state, the goal of amplitude estimation is to estimate the norm of the projection of the state vector on the sub-space spanned by good states. According to the basic properties of quantum mechanics, the square of the vector norm equals the probability that a good state is obtained if the quantum state is measured. Let $p$ denote this probability and $\sqrt{p}$ denote the corresponding vector norm, i.e., amplitude.
 Amplitude estimation is different from another problem -- quantum state estimation, in which the goal is to reconstruct the entire pure or mixed quantum states based upon measurements on copies of identical quantum states.   The readers are referred to \cite{gill_jrssb,gill2008conciliation,gill2013asymptotic} for statistical methods on this problem. Amplitude estimation, by contrast, focuses on a single parameter $p$. 
Amplitude estimation has various applications, e.g., in finance \citep{rebentrost2018quantum,zoufal2019quantum,woerner2019quantum,egger2020credit}, chemistry \citep{knill2007optimal,kassal2008polynomial}, machine learning \citep{wiebe2015quantum,wiebe2016quantum}, and generic tasks such as Monte
Carlo sampling \citep{montanaro2015quantum} and numerical integration \citep{montanaro2015quantum,suzuki2020amplitude}.


\cite{brassard2002quantum} formulated  amplitude estimation  as a quantum phase estimation (QPE) problem  \citep{kitaev1995quantum} and proved that  QPE-based amplitude estimation can achieve a quadratic speedup over classical Monte Carlo sampling -- that is,  the number of oracle queries achieves $O(1/\epsilon)$ where $\epsilon$ is the desired level of precision.  \cite{suzuki2020amplitude} mentioned that  QPE-based amplitude estimation involves many controlled operations, i.e., controlled Grover operators,  that can be difficult
to implement on noisy
intermediate-scale quantum (NISQ) devices. In addition, QPE-based amplitude estimation relies on quantum Fourier transform (QFT), as mentioned by \cite{aaronson2020quantum},   more commonly associated with Shor's algorithm that can achieve an exponential speedup. This raises a natural question of whether one can design an amplitude estimation algorithm, which is based only on Grover iterations and can achieve quadratic speedup. A number of Grover-based amplitude estimation algorithms have recently been proposed. \cite{suzuki2020amplitude} built a maximum likelihood estimate for the amplitude, thereafter called maximum likelihood amplitude estimation (MLAE), based on samples generated from Grover's algorithm with various numbers of iterations. The paper  provides a lower bound of the estimation error. \cite{wie2019simpler} replaced QPE by the Hadamard test in the proposed algorithm. \cite{aaronson2020quantum} proposed the first Grover-based amplitude estimation algorithm with the theoretically guaranteed quadratic speedup. However, the constants in the theoretical bound are very large and the empirical estimation error is also large for practical usage \citep{grinko2021iterative}. \cite{grinko2021iterative} proposed an iterative algorithm for quantum amplitude estimation, thereafter called IQAE, and provided a proof of the correctness of the algorithm and the quadratic speedup up to a double-logarithmic factor. \cite{grinko2021iterative} included a sub-routine FINDNEXTK to search for search for the appropriate number of Grover iterations, which can be time-consuming.
\cite{nakaji2020faster} recently proposed another Grover-based algorithm with a theoretical guarantee, but the empirical estimation error appears to be substantially larger than MLAE and IQAE (see Figure 3 of \cite{nakaji2020faster} and Figure 3 of \cite{grinko2021iterative}). 

We propose a new Grover-based algorithm for amplitude estimation, called adaptive algorithm.  The amplitude cannot be uniquely identified from the measurements if only a single circuit of Grover iterations is used \citep{suzuki2020amplitude}, which creates a unique challenge -- \textit{period ambiguity} in estimation. We design an adaptive algorithm that gradually increases the number of Grover iterations such that the confidence interval\footnote{This confidence interval is in fact for $\theta=\arcsin \sqrt{p}$. The reason of introducing this reparametrization is given in Section \ref{sec:amplitude_grover}.} in each step can be uniquely determined based on the period estimated from the previous steps. In particular, we introduce an adjustment factor which adjusts the probability of obtaining a good state, and hence the amplitude, when the interval's length does not exceed the period's length but the interval overlaps with two periods. We show that the amplitude after the adjustment, and hence  the original amplitude, can be estimated without ambiguity in the subsequent step. With this adjustment, our algorithm does not rely on a search sub-routine as in  \cite{grinko2021iterative}, which can be time-consuming for certain parameter values. Moreover, the number of total steps and the number of measurements in each step are easier to bound analytically. We therefore give a rigorous proof of the correctness of the algorithm and the quadratic speedup up to a double-logarithmic factor. Furthermore, we show with numerical studies that the proposed algorithm uses a similar number of quantum queries to achieve the same level of precision $\epsilon$ compared to MLAE and IQAE, but the classical part, i.e., the non-quantum part has substantially lower computational complexity. A simple analysis shows that the computational complexity of the classical part achieves $O(\log (1/\epsilon) \log(\log(1/\epsilon)))$.

We summarize the contributions  as follows:
\begin{itemize}
    \item We introduce a novel variant of interval estimation for quantum amplitudes based on Grover's algorithm. One of the key ingredients is an adaptive adjustment factor.
    \item The new algorithm is easier for theoretical analysis and we prove that the  number of \textit{oracle queries} achieves $O((1/\epsilon) \log(\log(1/\epsilon)))$, which is a quadratic speedup over  classical Monte Carlo sampling up to a double-logarithmic factor.
        \item The computational complexity of the \textit{classical part} of the algorithm achieves $O(\log (1/\epsilon) \log(\log(1/\epsilon)))$. We show by numerical studies that the classical part  has   substantially lower  computational costs than state-of-the-art algorithms. 
\end{itemize}

The remainder of this article is organized as follows. We give a brief background review on quantum computing and on the framework for Grover-based amplitude estimation in Section \ref{sec:prelim}. We explain the main idea of the algorithm in Section \ref{sec:idea} and present the algorithm and the theoretical analysis in Section \ref{sec:alg}. The numerical comparison with state-of-the-art algorithms is given in Section \ref{sec:numerical}. We conclude the paper by a discussion on future research problems in Section \ref{sec:conclusion}.

\section{Preliminary}\label{sec:prelim}
\subsection{Brief background review on quantum computing}
A quantum bit or qubit is the quantum version of the classic  bit. The quantum state of a qubit is represented by a linear combination, or called superposition, of two orthonormal basis states. That is,
\begin{align*}
\ket{\xi}= \begin{pmatrix}
\alpha_0 \\
\alpha_1
\end{pmatrix} = \alpha_0 \ket{0}+\alpha_1 \ket{1},
\end{align*}
where  $\ket{0}$ and $\ket{1}$ are the basis states:
\begin{align*}
    \ket{0}=\begin{pmatrix}
1 \\
0 
\end{pmatrix}, \,\, 
\ket{1}=\begin{pmatrix}
0 \\
1 
\end{pmatrix},
\end{align*}
and $\alpha_0, \alpha_1$ are complex numbers, called amplitudes, satisfying $|\alpha_0|^2 +|\alpha_1|^2=1$.
The notation $\ket{\cdot}$, called ``ket'', denotes a column vector, and $\bra{\cdot}$, called ``bra'', denotes the conjugate transpose of the corresponding $\ket{\cdot}$.

A basis state of $n$ multiple qubits has the form $\ket{x_1 x_2 ... x_n} = \ket{x_1}\otimes \ket{x_2} \otimes \cdots \otimes \ket{x_n}$, where $\otimes$ is the Kronecker product and $x_i=0$ or 1 for $i=1,...,n$. The notation $\otimes$ is usually omitted, i.e., $\ket{x_1}\ket{x_2}\cdots\ket{x_n}= \ket{x_1}\otimes \ket{x_2} \otimes \cdots \otimes \ket{x_n}$. For example, $\ket{00}$, $\ket{0}\ket{0}$ and $\ket{0}\otimes\ket{0}$ are the same. 

The state of $n$ multiple qubits is represented by a unit vector in $\mathbb{C}^{2^n}$ with the form
\begin{align*}
    \ket{\xi}=\sum_{x\in \{0,1\}^n} \alpha_{x} \ket{x}.
\end{align*}
One important feature in quantum computing is that we cannot acquire the values of the amplitudes of a quantum state directly (\cite{nielsen}, Section 1.2). Instead, we can only acquire information from a quantum state through \textit{measurement}. Specifically, $\ket{x}$ is obtained with probability $|\alpha_{x}|^2$ when $\ket{\xi}$ is measured.

A quantum state can be changed by unitary transformations. A unitary transformation on an $n$-qubit state can be represented by a $2^n \times 2^n$ unitary matrix. The design of  useful unitary transformations is the heart of quantum computing. 

\subsection{Amplitude estimation based on Grover's algorithm}\label{sec:amplitude_grover}

 The quantum amplitude estimation problem was first introduced by  \cite{brassard2002quantum}. We follow the description\footnote{The original formulation in \cite{brassard2002quantum} does not include the ancilla bit.} in \cite{suzuki2020amplitude} and \cite{grinko2021iterative}. Consider the $2^{n+1}$ basis states of $n+1$ qubits. Define the  basis states with the last qubit on $\ket{1}$ as \textit{good} states and those with the last qubit on $\ket{0}$ as \textit{bad} states.  Let $\mathcal{A}$ be a unitary transformation on $n+1$ qubits and
$\ket{\Psi}= \mathcal{A}\ket{0}_{n+1}$.
 Write $\ket{\Psi}$ as a linear combination of the  basis states:
 \begin{align*}
     \ket{\Psi} = \sum_{x \in \{0,1\}^n} \alpha_{x,1} \ket{x}\ket{1}+\sum_{x \in \{0,1\}^n} \alpha_{x,0} \ket{x}\ket{0}.
 \end{align*}
 When the last qubit $\ket{\Psi}$ of  is measured, $\ket{1}$ is obtained with probability $\sum_{x \in \{0,1\}^n} |\alpha_{x,1}|^2$ according to the basic properties of quantum computing \citep{nielsen}. Let $p$ denote this probability. The goal of amplitude estimation is to estimate $p$.

Let $\ket{\Psi_1} =(1/\sqrt{p})\sum_{x \in \{0,1\}^n} \alpha_{x,1} \ket{x}$ and  $\ket{\Psi_0}=(1/\sqrt{1-p})\sum_{x \in \{0,1\}^n} \alpha_{x,0} \ket{x}$. $\ket{\Psi}$ can be written as 
\begin{align}
 \ket{\Psi}=\sqrt{p} \ket{\Psi_1}\ket{1}+\sqrt{1-p} \ket{\Psi_0} \ket{0}. \label{A_def}
\end{align}
In the following, $\ket{\Psi_1}\ket{1}$ and $\ket{\Psi_0}\ket{0}$ are called normalized good and bad states, respectively. Note that $\ket{\Psi_1}$ and  $\ket{\Psi_0}$ are not necessarily orthogonal,  and with the ancilla bit, $\ket{\Psi_1}\ket{1}$ and $\ket{\Psi_0} \ket{0}$ are orthogonal. 

A special case of $\mathcal{A}$ corresponds to querying Boolean functions through quantum oracles. Let $f: \{0,1\}^n \rightarrow \{0,1\}$ be a Boolean function. One can query $f$ with a quantum oracle in the form of a unitary transformation $\mathcal{U}_f$ defined as
\begin{align*}
\mathcal{U}_f\ket{x}\ket{y} =\ket{x}\ket{y\oplus f(x)},
\end{align*}
where $x\in \{0,1\}^n$, $y \in \{0,1\}$, and $\oplus$ is the modulo 2 addition. The beauty of $\mathcal{U}_f$ is that it allows quantum computers to evaluate $f(x)$ for all $2^n$ values of $x$ simultaneously (\cite{nielsen}, Section 1.4.2).  Let $\mathcal{H}$ be the Hadamard transform on one qubit, that is,
\begin{align*}
\mathcal{H} = \frac{1}{\sqrt{2}} \begin{pmatrix}
1 & 1 \\
1 & -1
\end{pmatrix}.
\end{align*}
Let $\mathcal{H}^{\otimes n}$ be the Kronecker product of $n$ Hadamard transforms, which changes $\ket{0}_n$ to the uniform superposition:
\begin{align*}
\mathcal{H}^{\otimes n} \ket{0}_n = \frac{1}{\sqrt{2^n}} \sum_{x \in \{0,1\}^n } \ket{x}.
\end{align*} 
One can define\footnote{We follow the notation convention in the quantum computing literature, for example, \cite{aaronson2020quantum}: $I_n$  is the identity matrix on $n$ qubits, that is, a $2^n \times 2^n$ matrix.} $\mathcal{A}$ as $\mathcal{A}=\mathcal{U}_f (\mathcal{H}^{\otimes n} \otimes I_1)$, which has the form in \eqref{A_def}:
\begin{align*}
 \mathcal{A}\ket{0}_{n+1} = \sqrt{p} \frac{1}{\sqrt{\#\{x: f(x)=1\}}} \sum_{x: f(x)=1} \ket{x}\ket{1} + \sqrt{1-p} \frac{1}{\sqrt{\#\{x: f(x)=0\}}} \sum_{x: f(x)=0} \ket{x}\ket{0},
\end{align*}
where $p$ is the proportion of $x$ in $\{0,1\}^n$ such that $f(x)=1$. 

If one estimates $p$ by classical Monte Carlo sampling, that is,  sampling $x_1,...,x_M$ independently and uniformly from $\{0,1\}^n$ and using $(1/M) \sum_{i=1}^M f(x_i)$ as the estimate, then the estimation error is $O(1/\sqrt{M})$. Here $M$ equals the number of times $f$ is queried. By contrast, the estimation error can achieve $O(1/M)$ using amplitude estimation, up to possible logarithmic factors, by querying $\mathcal{U}_f$  through a quantum computer for $M$ times.  

 We focus on amplitude estimation based on  amplitude amplification \citep{brassard2002quantum}, an algorithm that generalizes Grover's algorithm \citep{grover1996fast}.  We follow the description in  \cite{suzuki2020amplitude}. Instead of measuring $\ket{\Psi}$ directly, one can apply the following operator on $\ket{\Psi}$:
\begin{align*}
\mathcal{Q} = -\mathcal{A} \mathcal{S}_0 \mathcal{A}^{-1} \mathcal{S}_{\chi}, 
\end{align*}
where
\begin{align}
\mathcal{S}_0 = I_{n+1}-2 \ket{0}_{n+1} \bra{0}_{n+1}, \,\,
\mathcal{S}_{\chi} =  I_{n+1} - 2 (I_n \otimes \ket{1} \bra{1}).  \label{chi}
\end{align}
In the following, $\mathcal{Q}$ is referred to as the Grover operator. The operator $-\mathcal{A} \mathcal{S}_0 \mathcal{A}^{-1}$  performs  a reflection with respect to $\ket{\Psi}$. And the operator $\mathcal{S}_{\chi}$ puts a negative sign to good states and does nothing to  bad states, that is, $\mathcal{S}_{\chi}\ket{\Psi}=-\sqrt{p} \ket{\Psi_1}\ket{1}+\sqrt{1-p} \ket{\Psi_0} \ket{0}$. Also note that $\mathcal{S}_{\chi}$ in \eqref{chi} identifies good states by simply checking whether the last qubit is on $\ket{1}$. 

Let $\theta = \arcsin \sqrt{p}$, which is in $[0,\frac{\pi}{2}]$.  \cite{brassard2002quantum} showed that applying $\mathcal{Q}$ on $\ket{\Psi}$  for $m$ times gives
\begin{align*}
\mathcal{Q}^m \ket{\Psi} = \sin ((2m+1)\theta) \ket{\Psi_1}\ket{1}+ \cos ((2m+1)\theta) \ket{\Psi_0}\ket{0}, 
\end{align*}
which implies that one obtains $\ket{1}$ with probability $\sin^2 ((2m+1)\theta)$ when measuring the last qubit of $\mathcal{Q}^m\ket{\Psi}$. 
In general, one can select a sequence of $m_t$ values for $t=0,...,T$, and for each $m_t$ take $N_t$ independent measurements by repeating the above process for $N_t$ times.  Let $X_t$ be the number of good states among the $N_t$ measurements, which follows a binomial distribution:
\begin{align}
    \mathbb{P}\left (X_t =s  \right )   = \binom{N_t}{s} \left (\sin^2 ((2m_t+1)\theta) \right )^s \left (\cos^2 ((2m_t+1)\theta) \right )^{N_t-s}, \,\, s=0,...,N_t. \label{Q_m}
\end{align}
Define $N_\textnormal{oracle}=\sum_{t=0}^T N_t m_t$, called the number of oracle queries\footnote{Rigorously speaking, it seems more appropriate to count the number of oracle queries in one application of $\mathcal{Q}$ twice. Here we follow the definition in \cite{grinko2021iterative} for comparison. }, which measures the complexity of the sample in this scenario because one needs to apply $\mathcal{Q}$ for $m_t$ times to obtain a single measurement. The goal is to make the estimation error for $p$ achieve $O(1/N_\textnormal{oracle})$ up to a possible logarithmic factor.

\section{Main Idea}\label{sec:idea}

Eq. \eqref{Q_m} is the starting point of a number of recent Grover-based amplitude estimation methods \citep{aaronson2020quantum,suzuki2020amplitude,grinko2021iterative,nakaji2020faster}, including ours. 

The original motivation of   applying $\mathcal{Q}$ repeatedly is to increase the amplitude $\sqrt{p}$ approximately linearly for small $p$.  By contrast, a classical brute-force search algorithm increases the probability $p$ linearly. Amplitude amplification therefore  achieves a quadratic speedup over the classical brute-force search when $p$ is small. As discovered by \cite{aaronson2020quantum} and \cite{suzuki2020amplitude},  applying $\mathcal{Q}$ repeatedly also improves the estimation of $p$ despite that $p$ is not necessarily small.

The estimation error  based on a Monte Carlo sample of size $M$ scales as $O(1/\sqrt{M})$. By contrast, increasing $m_t$ in \eqref{Q_m} in an appropriate manner can reduce the estimation error to $O(1/m_t)$. We briefly explain the reason.  Let $[L,U]$ be a confidence interval for $\sin^2 ((2m_t+1)\theta)$ based on $X_t$ and $N_t$. Due to the periodicity of $\sin^2 ((2m_t+1)\theta)$, such an interval is equivalent to the union of $2m_t+1$ intervals for $\theta$:
\begin{align}
I ^{+, (j)} & =  \left [ \frac{\arcsin \sqrt{L}+j \pi}{2m_t+1}, \frac{\arcsin \sqrt{U}+j \pi}{2m_t+1} \right ], \quad j=0,1,...,m_t,  \nonumber \\
I ^{-, (j)} & =   \left [ \frac{-\arcsin \sqrt{U}+j \pi}{2m_t+1}, \frac{-\arcsin \sqrt{L}+j \pi}{2m_t+1} \right ], \quad j=1,...,m_t. \label{interval_pn}
\end{align}
 Note that each interval is contained in one of the intervals $[0,\frac{1}{2m_t+1}\frac{\pi}{2}],[\frac{1}{2m_t+1}\frac{\pi}{2},\frac{2}{2m_t+1}\frac{\pi}{2}],..., [\frac{2m_t}{2m_t+1}\frac{\pi}{2},\frac{\pi}{2}]$, referred to as \textit{period} in the following. If we are able to determine the correct period, then the estimation error for $\theta$ is in the order of $O(1/m_t)$. The estimation error for $p$ is also in the order of $O(1/m_t)$ since $p=\sin^2(\theta)$ is Lipschitz continuous.

It is a natural idea to design a sequential algorithm to determine the period. First, use the measurements from the original $\ket{\Psi}$, i.e., $m_0=0$, to construct an initial confidence interval for $\theta$, which does not have the multi-value issue. In the following steps, use the confidence interval estimated in the previous step to determine the period of $\theta$. For simplicity, assume for now that $m_t$ grows at a geometric rate such as $2m_t+1=K^t$, where $K$ is an odd number.  The number of oracle queries  in the final step is therefore in the same order of that in all previous steps, both $O(K^T)$. This implies the estimation error is in the order of $O(1/N_\textnormal{oracle})$.   Similar ideas have appeared in the literature:  although not design a sequential algorithm, \cite{suzuki2020amplitude} recommended using an exponentially incremental sequence $\{m_t\}$ in MLAE. \cite{grinko2021iterative} designed a sequential algorithm that uses  data-dependent $\{m_t\}$ determined by a search sub-routine, which will be discussed  in Section \ref{sec:numerical}.  Our choice of $\{m_t\}$ will be given in Section \ref{sec:alg}.

Although promising, the above idea  has a serious caveat: when the true value of $\theta$ is at or near the boundary of two adjacent periods for the subsequent step, the estimated confidence interval can  overlap with both periods even though we are able to control the length of the interval.  Panel (a) in Figure \ref{fig:idea2} illustrates this problem. The true value of $p$ is set as 0.2 in this toy example. The confidence interval for $\theta$ from step 0 overlaps with two periods $[0,\pi/6]$ and $[\pi/6,\pi/3]$, which brings difficulty in step 1: the algorithm does not know how to choose between the two intervals, each contained in a period. We propose the following solution to the problem, which is the key ingredient of our algorithm. 
\begin{figure}[!htp]
	\twoImages{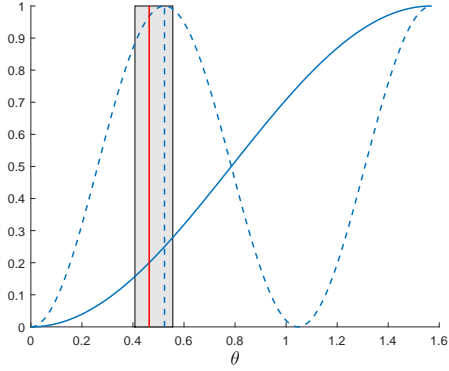}{7cm}{(a)  Confidence interval for $\theta$ in step 0}{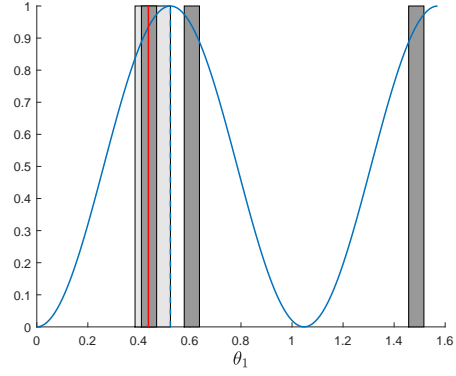}{7cm}{(c) Confidence interval for $\theta_1$ in step 1}
	\twoImages{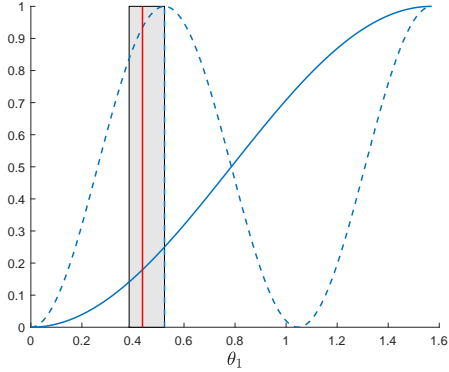}{7cm}{(b)  Confidence interval for $\theta_1$ in step 0}{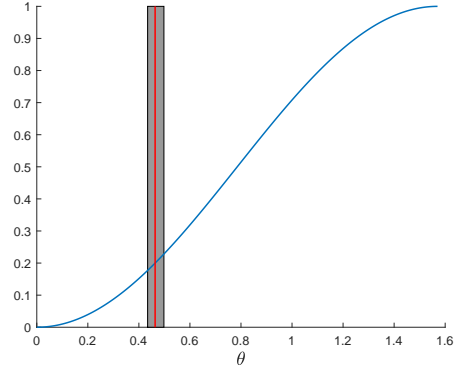}{7cm}{(d) Confidence interval for $\theta$ in step 1}
	\caption{An illustration of the first two steps of the adaptive algorithm. Panel (a): the confidence interval for $\theta$ in step 0 overlaps with two periods $[0,\frac{\pi}{6}]$ and $[\frac{\pi}{6},\frac{\pi}{3}]$. Panel (b): the upper limit of the confidence interval for $\theta_1=\arcsin \sqrt{r_1 \sin^2(\theta)}$  coincides with the boundary so that the interval is contained in the period $[0,\frac{\pi}{6}]$. Panel (c): the confidence interval for $\theta_1$ in step 1, constructed by running the Grover operator \eqref{adjust_grover} on  $(\mathcal{A}\otimes \mathcal{R}_{t+1})\ket{0}_{n+1}\ket{0}$. Panel (d): convert to the confidence interval for $\theta$.  }
	
	\label{fig:idea2}
\end{figure}

Denote the confidence interval for $\theta$ in step $t$ by  $[\theta_{t}^L,\theta_{t}^U]$. If there exists $k$ such that $\frac{k}{2m_{t+1}+1} \frac{\pi}{2} \in (\theta_{t}^L,\theta_{t}^U)$, then we introduce an adjustment factor
\begin{align*}
r_{t+1}  = \left . \sin ^2 \left (\frac{k}{2m_{t+1}+1} \frac{\pi}{2} \right ) \right / \sin^2 (\theta_{t}^U);
\end{align*}
otherwise $r_{t+1}=1$. The factor $r_{t+1}$ adjusts the scale of $\theta$ and makes the new confidence interval contained in a single  period of length $\frac{1}{2m_{t+1}+1} \frac{\pi}{2}$. Specifically, let $\theta_{t+1}=\arcsin \sqrt {r_{t+1} p }$. It is easy to check that the upper limit of the confidence interval for $\theta_{t+1}$ is  $\frac{k}{2m_{t+1}+1} \frac{\pi}{2}$ according to the definition of $r_{t+1}$. We will prove in Lemma \ref{thm:shorten} that the confidence interval contracts with this adjustment, which implies that the interval for $\theta_{t+1}$ is fully contained in  $[\frac{k-1}{2m_{t+1}+1} \frac{\pi}{2}, \frac{k}{2m_{t+1}+1} \frac{\pi}{2}]$ when we control the length of $[\theta_{t}^L,\theta_{t}^U]$ within $\frac{1}{2m_{t+1}+1}\frac{\pi}{2}$. 

In step $t+1$, we add a qubit and define an adjusted unitary transformation  on  $n+2$ qubits such that the probability of obtaining a good state is $r_{t+1}p$ when measuring the state after the  transformation. Such an adjustment has been introduced in \cite{aaronson2020quantum} for a different purpose. Let $\mathcal{R}_{t+1}$ satisfy $\mathcal{R}_{t+1}\ket{0} =\sqrt{r_{t+1}}\ket{1}+\sqrt{1-r_{t+1}} \ket{0}$. Then
 \begin{align}
\ket{\tilde{\Psi}}=(\mathcal{A}\otimes \mathcal{R}_{t+1})\ket{0}_{n+1}\ket{0} & =\sqrt{r_{t+1}p} \ket{\Psi_1}\ket{1}\ket{1}  +\sqrt{r_{t+1}(1-p)} \ket{\Psi_0} \ket{0}\ket{1} \nonumber \\
& \quad +\sqrt{(1-r_{t+1})p} \ket{\Psi_1} \ket{1}\ket{0}+\sqrt{(1-r_{t+1})(1-p)} \ket{\Psi_0} \ket{0}\ket{0}. \label{adjust_A}
\end{align}
Now a basis state is defined as good state if  the last two qubits are on $\ket{11}$. 

\begin{remark}
\textit{A quantum state can be viewed as a random variable when it is measured. Using the terminology of probability, the above operation is  adding a ``random variable'' $\mathcal{R}_{t+1}\ket{0}$, which equals 1 with probability $r_{t+1}$. From \eqref{adjust_A}, $\ket{\Psi}$ and $\mathcal{R}_{t+1}\ket{0}$  can be understood as ``independent random variables''. The ``joint probability'' of both $\mathcal{R}_{t+1}\ket{0}$ and the last bit of $\ket{\Psi}$  being 1 is therefore $r_{t+1}p$. Below we give details of the Grover operator on $ \ket{\Psi} \otimes \mathcal{R}_{t+1}\ket{0}$. The readers who are only interested in the statistical model can directly go to \eqref{adjusted_binom}.}
\end{remark}

Note that the last three terms in \eqref{adjust_A} are orthogonal to $\ket{\Psi_1 1 1}$. Denote the combination of these terms, after normalization, by $\ket{\Psi_111^{\perp}}$. We define an operator that amplifies the amplitude $\sqrt{r_{t+1}p}$:
\begin{align}
\mathcal{Q}_{t+1} = -(\mathcal{A}\otimes \mathcal{R}_{t+1}) (I_{n+2}-2 \ket{0}_{n+2} \bra{0}_{n+2}) (\mathcal{A}\otimes \mathcal{R}_{t+1})^{-1} ( I_{n+2} - 2 (I_n \otimes \ket{11} \bra{11})). \label{adjust_grover}
\end{align} 
In \eqref{adjust_grover}, the operator $I_{n+2} - 2 (I_n \otimes \ket{11} \bra{11})$ puts a negative sign to $\ket{\Psi_1 1 1}$ and does nothing to  other terms. That is,
\begin{align*}
( I_{n+2} - 2 (I_n \otimes \ket{11} \bra{11})) \ket{\Phi}  & =-\sqrt{r_{t+1}p} \ket{\Psi_1 1 1}  +\sqrt{1-r_{t+1}p} \ket{\Psi_111^{\perp}}.
\end{align*}
The operator $-(\mathcal{A}\otimes \mathcal{R}_{t+1}) (I_{n+2}-2 \ket{0}_{n+2} \bra{0}_{n+2}) (\mathcal{A}\otimes \mathcal{R}_{t+1})^{-1}=-(I_{n+2}-2 \ket{\tilde{\Psi}} \bra{\tilde{\Psi}})$ performs  a reflection with respect to $\tilde{\Psi}$. Therefore, by the same argument in \cite{brassard2002quantum},
\begin{align}
\mathcal{Q}_{t+1}^{m_{t+1}} \ket{\tilde{\Psi}} = \sin ((2m_{t+1}+1)\theta_{t+1}) \ket{\Psi_1 11}+ \cos ((2m_{t+1}+1)\theta_{t+1}) \ket{\Psi_111^{\perp}}. \label{adjust_Q_m}
\end{align}
When repeating the process and measuring $\mathcal{Q}_{t+1}^{m_{t+1}} \ket{\tilde{\Psi}}$ for $N_{t+1}$ times, the number of observed good states $X_{t+1}$ follows a binomial distribution:
\begin{align}
        \mathbb{P}\left (X_{t+1} =s  \right )   = \binom{N_{t+1}}{s} \left (\sin^2 ((2m_{t+1}+1)\theta_{t+1}) \right )^s \left (\cos^2 ((2m_{t+1}+1)\theta_{t+1}) \right )^{N_{t+1}-s}, \,\, s=0,...,N_{t+1}. \label{adjusted_binom}
\end{align}
Since $\theta_{t+1}$ is contained in  $[\frac{k-1}{2m_{t+1}+1} \frac{\pi}{2}, \frac{k}{2m_{t+1}+1} \frac{\pi}{2}]$ for  certain $k$, only one interval with the form in \eqref{interval_pn} is a legitimate confidence interval for  $\theta_{t+1}$. Finally, we convert the interval for $\theta_{t+1}$ back to the interval for $\theta$. Denote the new interval for $\theta$ by $[\theta_{t+1}^L,\theta_{t+1}^U]$. At the same time, we select appropriate $N_{t+1}$ such that  $| \theta_{t+1}^L-\theta_{t+1}^U  |\leq \frac{1}{2m_{t+2}+1}\frac{\pi}{2}$ so that the recursion can continue. We illustrate the first two steps of the above procedure in Figure \ref{fig:idea2}.

We apply \eqref{adjust_A} and \eqref{adjust_grover} in a different way than \cite{aaronson2020quantum}. In their method,  $\mathcal{R}$ is defined as $\mathcal{R}\ket{0} =\frac{1}{1000}\ket{1}+\sqrt{1-\frac{1}{1000^2}} \ket{0}$, which shrinks $p$ by a factor of $10^{-6}$. By contrast, we adjust $r_t$ adaptively to avoid the period ambiguity in each step. In practice, $r_t$ is typical close to 1, which makes the estimation lose very little efficiency due to the adjustment. 

\section{Algorithm}\label{sec:alg}
We formally describe the adaptive algorithm in Algorithm \ref{alg:adaptive}. Without loss of generality, we assume $p \in [0,\frac{1}{2}]$ because otherwise one can add artificial bad states to the system by adding a qubit on $\frac{1}{\sqrt{2}}\ket{1}+\frac{1}{\sqrt{2}} \ket{0}$ at the beginning of the algorithm. We need such an assumption to control the length  of the confidence interval when we convert the  interval for $\theta_t$ back to the interval for $\theta$ (Panel (d) in Figure \ref{fig:idea2}). See lines 17--19 in Algorithm \ref{alg:adaptive} and Lemma \ref{thm:stretch} for details.  

\begin{algorithm}[!htp]
	\KwIn{$\epsilon$, $\alpha$, $K$,  $N_{\textnormal{shots}}$ \tcp{$\epsilon$: precision level; $1-\alpha$: confidence level; $K$: odd number $\geq 3$}}
	 $T \leftarrow  \left \lceil \left. \log  \frac{\pi}{K \epsilon} \right / \log K \right \rceil,r_{0}\leftarrow1, \hat{k}_{0}\leftarrow0, m_0\leftarrow 0$ \tcp{$T$: upper bound of the number of iterations} 
  \For{$t=0$ to $T$}{
add artificial bad states as  \eqref{adjust_A} such that  $\mathbb{P}(\textnormal{good state})=r_tp$ \\
$X_t \leftarrow 0,j\leftarrow 0$ \\
	\Repeat{$|\theta^L_t-\theta^U_t | \leq \frac{1}{K} \frac{1}{2m_t+1}\frac{\pi}{2}$}{
		\tcp{increase the sample size $N_t$ by $N_\textnormal{shots}$ at each time until the length of the CI for $\theta$ is less than or equal to $\frac{1}{K} \frac{1}{2m_t+1}\frac{\pi}{2}$}
		$j\leftarrow j+1$ \\
		$N_t \leftarrow j N_{\textnormal{shots}}$ \\
		$X_t \leftarrow$ $X_t+$ the number of good states observed by measuring $\mathcal{Q}_{t}^{m_{t}} \ket{\tilde{\Psi}}$ (defined in \eqref{adjust_Q_m}) for $N_{\textnormal{shots}}$ times \\
		$\delta_t\leftarrow \sqrt{\log \left ( \frac{\pi^2    (T+1)}{3\alpha } j^2 \right ) \frac{1}{2N_t} }$ \tcp{the choice of $\delta_t$ makes the coverage probability  at least $1-\frac{\alpha}{T+1}$ in step $t$}
		$L_t  \leftarrow \max \left (\frac{X_t}{N_{t}}-\delta_t, 0 \right ), U_t \leftarrow \min \left (\frac{X_t}{N_{t}}+\delta_t, 1 \right )$ \\
		  \uIf{$\hat{k}_{t} \equiv 0\, (\textnormal{mod } 2)$ \tcp{CI for $\theta_t$, based on the period estimated from step $t-1$}}{
$[\theta^L_t,\theta^U_t] \leftarrow \left [\frac{\arcsin \sqrt{L_t}+\hat{k}_{t}\pi /2}{2m_t+1}, \frac{\arcsin \sqrt{U_t}+\hat{k}_{t}\pi/2}{2m_t+1} \right ]$
		}
		\Else{
$[\theta^L_t,\theta^U_t] \leftarrow \left [\frac{-\arcsin \sqrt{U_t}+(\hat{k}_{t}+1)\pi/2}{2m_t+1},  \frac{-\arcsin \sqrt{L_t}+(\hat{k}_{t}+1)\pi/2}{2m_t+1} \right ]$
		}
	$\theta_{t}^L\leftarrow \min(\theta_{t}^L,\arcsin \sqrt{r_{t}/2}),\theta_{t}^U\leftarrow \min(\theta_{t}^U,\arcsin \sqrt{r_{t}/2})$ \\
			  \uIf{$r_t<1$}{$[\theta^L_t,\theta^U_t] \leftarrow \left [  \arcsin \sqrt{(\sin \theta_{t}^L)^2/r_{t}},   \arcsin \sqrt{(\sin \theta_{t}^U)^2/r_{t}} \right ]$ \tcp{CI for $\theta$}}
} 
 $p^L\leftarrow \sin^2 (\theta_t^L),p^U \leftarrow \sin^2 (\theta_t^U)$ \\
 \uIf{$t=T$ OR $|p^L-p^U|\leq \epsilon$}{go to \textbf{Output}} 
 $m_{t+1} \leftarrow \left \lfloor \frac{1}{|\theta^L_t-\theta^U_t|}\frac{\pi}{4}-\frac{1}{2} \right \rfloor$ \\
 $\hat{k}_{t+1} \leftarrow \lfloor \left . 2 (2m_{t+1}+1) \theta_t^L \right / \pi \rfloor$ \tcp{determine the period} 
 \uIf{$\frac{\hat{k}_{t+1}+1}{2m_{t+1}+1}  \frac{\pi}{2}< \theta_t^U$}{$r_{t+1}  \leftarrow \left . \sin ^2 (\frac{\hat{k}_{t+1}+1}{2m_{t+1}+1} \frac{\pi}{2}) \right / \sin^2 (\theta_{t}^U)$ \tcp{set the adjustment factor $r_{t+1}$ if $(\theta^L_t,\theta^U_t)$ contains a boundary point of two periods}}
 \Else{$r_{t+1} \leftarrow 1$}
}
	\KwOut{$[p^L,p^U]$}
		\caption{Adaptive  Algorithm for Amplitude Estimation} \label{alg:adaptive}
	
\end{algorithm}

We choose $\{m_t\}$ that grows at least as fast as a geometric progression. That is, we choose  
$m_{t+1}$ as the largest integer such that  $|\theta^L_t-\theta^U_t | \leq  \frac{1}{2m_{t+1}+1}\frac{\pi}{2}$ (lines 20 and 24 in Algorithm \ref{alg:adaptive}), which implies $2m_{t+1}+1\geq K^t$. This choice takes full advantage of the precision of the current  interval for $\theta$ and can potentially make the length of the interval reach the desired precision level $\epsilon$ in fewer steps. 

Another  ingredient of Algorithm \ref{alg:adaptive} is that instead of preselecting the sample size $N_t$ in step $t$ such that $|\theta^L_t-\theta^U_t | \leq \frac{1}{K} \frac{1}{2m_t+1}\frac{\pi}{2}$, which is usually very conservative, we gradually increase $N_t$ by a fixed $N_\textnormal{shots}$ at each time until $|\theta^L_t-\theta^U_t | \leq \frac{1}{K} \frac{1}{2m_t+1}\frac{\pi}{2}$ is satisfied. This brings a subtle difficulty to the theoretical analysis. That is, when the condition $|\theta^L_t-\theta^U_t | \leq \frac{1}{K} \frac{1}{2m_t+1}\frac{\pi}{2}$ is met, the  data used to construct the confidence interval $[L_t,U_t]$ for $\sin^2 ((2m_t+1)\theta_t)$ (line 10 in Algorithm \ref{alg:adaptive}), rigorously speaking, is no longer a random sample. More specifically,
\begin{align*}
\mathbb{P}\left (X_t =j  \left | |\theta^L_t-\theta^U_t |  \leq \frac{1}{K} \frac{1}{2m_t+1}\frac{\pi}{2} \right. \right )   \neq \binom{N_t}{j} \left (\sin^2 ((2m_t+1)\theta_t) \right )^j \left (\cos^2 ((2m_t+1)\theta_t) \right )^{N_t-j},
\end{align*}
so Hoeffding's inequality \citep{Hoeffding63} cannot be directly applied to interval estimation. To resolve this difficulty, we choose $\delta_t$ (line 9 in Algorithm \ref{alg:adaptive})  such that an infinite sequence of  confidence intervals based on $jN_\textnormal{shots} \,\, (j=1,2,...)$ observations \textit{simultaneously} contain $\sin^2 ((2m_t+1)\theta_t)$ with probability at least $1-\frac{\alpha}{T+1}$. Therefore, the interval satisfying $|\theta^L_t-\theta^U_t | \leq \frac{1}{K} \frac{1}{2m_t+1}\frac{\pi}{2}$ also contains $\sin^2 ((2m_t+1)\theta_t)$ with probability at least $1-\frac{\alpha}{T+1}$. 

Next, we present the main theorem showing that the output $[p^L,p^U]$ of Algorithm \ref{alg:adaptive} reaches the pre-specified confidence level $1-\alpha$ and  precision level $\epsilon$. Moreover, the number of oracle queries $N_{\textnormal{oracle}}$ achieves $O(1/\epsilon)$ up to  a double-logarithmic factor of $\epsilon$. 
 \begin{theorem}\label{thm:main}
	If $0\leq p\leq 1/2$, the output $[p^L,p^U]$ of Algorithm \ref{alg:adaptive} satisfies the following properties:
	\begin{enumerate}
		\item $\mathbb{P}\left (p \in [p^L,p^U] \right ) \geq 1-\alpha$.
		\item $|p^L-p^U| \leq \epsilon$.
		\item $$N_{\textnormal{oracle}}=O \left ( \log \left ( \frac{\pi^2 (T+1) }{3 \alpha}  \right ) \frac{1}{\epsilon} \right ), $$
		where $T= \left \lceil \left. \log  \frac{\pi}{K \epsilon} \right / \log K \right \rceil$.
	\end{enumerate}
\end{theorem}  
\begin{proof}
	Proof of Claim 1: Let  $\theta=\arcsin \sqrt{p},\theta_t=\arcsin \sqrt {r_t p }, t=0,...,T$. We first show that 
	\begin{align*}
	\mathbb{P}\left ( \sin^2 ((2m_t+1) \theta_t) \in  [L_t, U_t], t=0,...,T   \right ) \geq 1-\alpha.
	\end{align*}
	Let $Z_{t,1},Z_{t,2},...$ be a sequence of independently and identically distributed random variables from $\textnormal{Ber}(\sin^2 ((2m_t+1) \theta_t))$. For $j=1,2,...$, let 
	\begin{align*}
	\eta_{j}= \sqrt{\log \left ( \frac{\pi^2    (T+1)}{3\alpha } j^2 \right ) \frac{1}{2j N_{\textnormal{shots}} }}.
	\end{align*} 
	From Hoeffding's inequality, for all $j$,
	\begin{align}
	& \mathbb{P} \left (\left. \left | \frac{1}{j N_{\textnormal{shots}}}\sum_{i=1}^{j N_{\textnormal{shots}} } Z_{t,i} -\sin^2 ((2m_t+1) \theta_t) \right | \geq \eta_j \right | m_t,r_t \right ) \leq 2\exp ( -2 j  N_{\textnormal{shots}} \eta_j^2 ) =\frac{\alpha}{T+1} \frac{6}{\pi^2} \frac{1}{j^2}, \nonumber \\
	& \mathbb{P} \left (\left. \exists j\in \{1,2,...\} \textnormal{ s.t.} \left | \frac{1}{j N_{\textnormal{shots}}}\sum_{i=1}^{j N_{\textnormal{shots}} } Z_{t,i} -\sin^2 ((2m_t+1) \theta_t) \right | \geq \eta_j \right | m_t,r_t \right ) \leq  \frac{\alpha}{T+1}  \sum_{j=1}^{\infty}\frac{6}{\pi^2} \frac{1}{j^2} = \frac{\alpha}{T+1}. \label{union_j}
	\end{align}
	Let $\hat{J}_t$ be the smallest integer in step $t$ such that $|\theta^L_t-\theta^U_t | \leq \frac{1}{K} \frac{1}{2m_t+1}\frac{\pi}{2}$, and let $N_t=\hat{J}_t N_{\textnormal{shots}}, \delta_t=\eta_{\hat{J}_t}$ (the \textbf{repeat} loop in Algorithm \ref{alg:adaptive}). We will leave until the proof of Claim 3 to show there is an upper bound for  $\hat{J}_t$.  
	Eq. \eqref{union_j} implies
	\begin{align*}
	\mathbb{P} \left (\left.   \left | \frac{1}{N_t}\sum_{i=1}^{N_t } Z_{t,i} -\sin^2 ((2m_t+1) \theta_t) \right | \geq \delta_t \right | m_t,r_t \right ) \leq \frac{\alpha}{T+1},
	\end{align*}
	and
	\begin{align*}
	\mathbb{P} \left (  \left | \frac{1}{N_t}\sum_{i=1}^{N_t } Z_{t,i} -\sin^2 ((2m_t+1) \theta_t) \right | \geq \delta_t   \right ) \leq \frac{\alpha}{T+1}.
	\end{align*}
	Therefore, 
	\begin{align*}
	\mathbb{P}\left ( \sin^2 ((2m_t+1) \theta_t) \in  [L_t, U_t], t=0,...,T   \right ) \geq 1-\alpha.
	\end{align*}
	Let $\hat{T}$ be the stopping time of $t$ in the algorithm. The above inequality implies
	\begin{align*}
	\mathbb{P}\left (  \sin^2 ((2m_t+1) \theta_t) \in  [L_t, U_t], t=0,...,\hat{T}  \right ) \geq 1-\alpha.
	\end{align*}
	
	Next we show that 
	\begin{align*}
	\sin^2 ((2m_t+1) \theta_t) \in [L_t, U_t], t=0,...,\hat{T} \quad (*)
	\end{align*}
	implies $p \in [p^L,p^U]$, which proves Claim 1. For the rest of the proof of Claim 1, we assume $(*)$. 
	
	We first show that for $t=0,...,\hat{T}$, $\theta_t$ belongs to the interval defined by (lines 11--16 in Algorithm \ref{alg:adaptive}):
	\begin{align*}
	& \left [\frac{\arcsin \sqrt{L_t}+\hat{k}_{t}\pi /2}{2m_t+1}, \frac{\arcsin \sqrt{U_t}+\hat{k}_{t}\pi/2}{2m_t+1} \right ], \,\, \textnormal{if $\hat{k}_t$ is even}, \\
	& \left [\frac{-\arcsin \sqrt{U_t}+(\hat{k}_{t}+1)\pi/2}{2m_t+1},  \frac{-\arcsin \sqrt{L_t}+(\hat{k}_{t}+1)\pi/2}{2m_t+1} \right ],  \,\, \textnormal{otherwise}.
	\end{align*}
	Denote the interval by $[\check{\theta}_t^L, \check{\theta}_t^U]$.
	
	We use induction. The conclusion holds for $t=0$ because 
	$\theta=\theta_0 \in  [\arcsin \sqrt{L_0},\arcsin \sqrt{U_0}]$, which is a single interval corresponding to $\hat{k}_0=0$. 
	Assume the conclusion holds for step $t$. We now consider step $t+1$. 
Let (lines 17--19)
	\begin{align*}
	\theta_t^L & = \arcsin \sqrt{ \sin^2 ( \min(\check{\theta}_{t}^L,\arcsin \sqrt{r_{t}/2} ) ) /r_{t}}=\min \left (\arcsin \sqrt{\sin^2 ( \check{\theta}_{t}^L) /r_{t}},\pi/4  \right ), \\
	\theta_t^U & =\arcsin \sqrt{ \sin^2 ( \min(\check{\theta}_{t}^U,\arcsin \sqrt{r_{t}/2} ) ) /r_{t}}=\min \left (\arcsin \sqrt{\sin^2 ( \check{\theta}_{t}^U) /r_{t}},\pi/4  \right ).
	\end{align*}
	Note that 
	\begin{align*}
	\theta_{t} \in [\check{\theta}_{t}^L, \check{\theta}_{t}^U] \Rightarrow \theta  \in \left [\arcsin \sqrt{\sin^2 ( \check{\theta}_{t}^L) /r_{t}},\arcsin \sqrt{\sin^2 ( \check{\theta}_{t}^U) /r_{t}} \right ],
	\end{align*}
	which further implies $\theta \in [\theta_{t}^L,\theta_{t}^U]$ since $p \leq 1/2$.
	
	Consider intervals $[0, \frac{1}{2m_{t+1}+1} \frac{\pi}{2}],[ \frac{1}{2m_{t+1}+1} \frac{\pi}{2},\frac{2}{2m_{t+1}+1} \frac{\pi}{2}],...,[\frac{2m_{t+1}}{2m_{t+1}+1} \frac{\pi}{2},\frac{\pi}{2}]$. The choice of $\hat{k}_{t+1}$ (line 25) makes $ \frac{\hat{k}_{t+1}}{2m_{t+1}+1} \frac{\pi}{2} \leq  \theta^L_t <\frac{\hat{k}_{t+1}+1}{2m_{t+1}+1} \frac{\pi}{2}$. If $\theta^U_t \leq \frac{\hat{k}_{t+1}+1}{2m_{t+1}+1} \frac{\pi}{2}$,
	$$[\theta^L_t,\theta^U_t] \subset \left [\frac{\hat{k}_{t+1}}{2m_{t+1}+1}\frac{\pi}{2},\frac{\hat{k}_{t+1}+1}{2m_{t+1}+1}\frac{\pi}{2} \right ].$$ Otherwise, define
	\begin{align*}
	r_{t+1}  = \left . \sin ^2 \left (\frac{\hat{k}_{t+1}+1}{2m_{t+1}+1} \frac{\pi}{2} \right ) \right / \sin^2 (\theta_{t}^U),
	\end{align*}
	which implies $\arcsin \sqrt{r_{t+1}(\sin\theta^U_t)^2}=\frac{\hat{k}_{t+1}+1}{2m_{t+1}+1} \frac{\pi}{2}$.

	From Lemma \ref{thm:shorten} in the appendix, since $r_{t+1}\leq 1$, 
	\begin{align*}
	\left |\arcsin \sqrt{r_{t+1}(\sin \theta^L_t)^2}-\arcsin \sqrt{r_{t+1}(\sin\theta^U_t)^2} \right |\leq  |\theta^L_t-\theta^U_t | \leq \frac{1}{2m_{t+1}+1}\frac{\pi}{2},
	\end{align*}
	where the last inequality is guaranteed by the algorithm (line 20). 
	Therefore, 
	\begin{align*}
	\left [ \arcsin \sqrt{r_{t+1}(\sin \theta^L_t)^2},\arcsin \sqrt{r_{t+1}(\sin\theta^U_t)^2} \right ] \subset \left  [\frac{\hat{k}_{t+1}}{2m_{t+1}+1}\frac{\pi}{2},\frac{\hat{k}_{t+1}+1}{2m_{t+1}+1}\frac{\pi}{2} \right ].
	\end{align*}
	
	Therefore, 
	\begin{align*} 
	& \theta \in [\theta_t^L ,\theta_t^U] \Rightarrow \theta_{t+1} \in \left [ \arcsin \sqrt{r_{t+1}(\sin \theta^L_t)^2},\arcsin \sqrt{r_{t+1}(\sin\theta^U_t)^2} \right ]  \\
	& \Rightarrow \theta_{t+1} \in \left [\frac{\hat{k}_{t+1}}{2m_{t+1}+1}\frac{\pi}{2},\frac{\hat{k}_{t+1}+1}{2m_{t+1}+1}\frac{\pi}{2} \right ].
	\end{align*}
	Note that  $\sin^2 ((2m_{t+1}+1)\phi)=\frac{1}{2} -\frac{1}{2}\cos(2(2m_{t+1}+1)\phi)$ is strictly increasing for all $\phi \in \left [\frac{\hat{k}_{t+1}}{2m_{t+1}+1}\frac{\pi}{2},\frac{\hat{k}_{t+1}+1}{2m_{t+1}+1}\frac{\pi}{2} \right ]$ when $\hat{k}_{t+1}$ is even, and is strictly decreasing in that interval when $\hat{k}_{t+1}$ is odd. When $\hat{k}_{t+1}$ is even, the unique solution of equation $\sin^2 ((2m_{t+1}+1) \phi)=y$ in $\left [\frac{\hat{k}_{t+1}}{2m_{t+1}+1}\frac{\pi}{2},\frac{\hat{k}_{t+1}+1}{2m_{t+1}+1}\frac{\pi}{2} \right ]$ is
	\begin{align}
	\phi=\frac{\arcsin \sqrt{y}}{2m_{t+1}+1}+\frac{\hat{k}_{t+1}}{2m_{t+1}+1}\frac{\pi}{2}. \label{solution}
	\end{align}
	In fact, one can verify that \eqref{solution} satisfies the equation and is within $\left [\frac{\hat{k}_{t+1}}{2m_{t+1}+1}\frac{\pi}{2},\frac{\hat{k}_{t+1}+1}{2m_{t+1}+1}\frac{\pi}{2} \right ]$. The solution is unique because the function is strictly monotonic. Therefore, the function $y=\sin^2 ((2m_{t+1}+1) \phi)$ has a inverse on $\left [\frac{\hat{k}_{t+1}}{2m_{t+1}+1}\frac{\pi}{2},\frac{\hat{k}_{t+1}+1}{2m_{t+1}+1}\frac{\pi}{2} \right ]$, defined by \eqref{solution}. Furthermore, since the function is  increasing,
	\begin{align*}
	& \,\, \theta_{k+1} \in \left [\frac{\hat{k}_{t+1}}{2m_{t+1}+1}\frac{\pi}{2},\frac{\hat{k}_{t+1}+1}{2m_{t+1}+1}\frac{\pi}{2} \right ] \textnormal{ and } \sin^2 ((2m_{t+1}+1) \theta_{t+1}) \in [L_{t+1}, U_{t+1}] \\
	\Leftrightarrow & \,\,\theta_{k+1} \in \left [\frac{\arcsin \sqrt{L_{t+1}}+\hat{k}_{t+1}\pi /2}{2m_{t+1}+1}, \frac{\arcsin \sqrt{U_{t+1}}+\hat{k}_{t+1}\pi/2}{2m_{t+1}+1} \right ].
	\end{align*}
	Similarly, when $\hat{k}_{t+1}$ is odd,
	\begin{align*}
	& \,\, \theta_{k+1} \in \left [\frac{\hat{k}_{t+1}}{2m_{t+1}+1}\frac{\pi}{2},\frac{\hat{k}_{t+1}+1}{2m_{t+1}+1}\frac{\pi}{2} \right ] \textnormal{ and } \sin^2 ((2m_{t+1}+1) \theta_{t+1}) \in [L_{t+1}, U_{t+1}] \\
	\Leftrightarrow & \,\,\theta_{k+1} \in \left [\frac{-\arcsin \sqrt{U_{t+1}}+(\hat{k}_{t+1}+1)\pi/2}{2m_{t+1}+1},  \frac{-\arcsin \sqrt{L_{t+1}}+(\hat{k}_{t+1}+1)\pi/2}{2m_{t+1}+1} \right ].
	\end{align*}
	We have therefore proved the conclusion for $t+1$. Moreover, we have also shown that $\theta \in [\theta_t^L ,\theta_t^U]$ in the above proof. Finally, by the definition of $[p^L,p^U]$, $p \in [p^L,p^U]$. 
	
	Proof of Claim 2: We only need to show that the claim holds if the algorithm stops at step $T$ because otherwise it automatically holds (line 22). Since $m_{t+1}$ is the largest integer such that  $|\theta^L_t-\theta^U_t | \leq  \frac{1}{2m_{t+1}+1}\frac{\pi}{2}$ (line 24) and the algorithm requires $|\theta^L_t-\theta^U_t | \leq \frac{1}{K} \frac{1}{2m_t+1}\frac{\pi}{2}$, we have $2m_{t+1}+1 \geq K(2m_t+1)$. A simple induction argument shows  $2m_t+1 \geq K^t$ for $t=0,...,T$.
	Since  
	\begin{align*}
	T= \left \lceil \left. \log  \frac{\pi}{K \epsilon} \right / \log K \right \rceil,
	\end{align*}
	we have 
	\begin{align*}
	& | \theta_T^L-\theta_T^U | \leq  \frac{1}{K} \frac{1}{2m_t+1}\frac{\pi}{2}   \leq   \frac{1}{K^{T+1}}  \frac{\pi}{2}, \\
	& | p^L-p^U| = \left | (\sin \theta^L_T)^2- (\sin \theta^U_T)^2  \right | \leq 2 | \theta_T^L-\theta_T^U | \leq  \frac{\pi}{K^{T+1}}  \leq \epsilon.
	\end{align*}	
	
	Proof of Claim 3: We first show that $\hat{J}_t$ has an upper bound.
	That is, if  
	\begin{align*}
	j=\max \left( \left \lceil \frac{4}{ c_t^2 N_{\textnormal{shots}}} \log \left ( \frac{\pi^2 (T+1) }{3 \alpha} \right )   \right \rceil, \left \lceil \frac{64 }{c_t^4 N_{\textnormal{shots}}^2 } \right \rceil \right),
	\end{align*}
	where $c_t=\sin^2 \left ( \sqrt{\frac{r_t}{2}} \frac{1}{K} \frac{\pi}{2} \right )$, and $N_t=j N_{\textnormal{shots}},\delta_t=\eta_j$, then we will  show $|\theta^L_t-\theta^U_t | \leq \frac{1}{K} \frac{1}{2m_t+1}\frac{\pi}{2}$.
	Recall that (line 9)
	\begin{align*}
	\delta_t= \sqrt{\log \left ( \frac{\pi^2    (T+1)}{3\alpha } j^2 \right ) \frac{1}{2j N_{\textnormal{shots}} }}.
	\end{align*}
	It follows that 
	\begin{align*}
	\delta_t & \,\, = \sqrt{ \left ( \log \left ( \frac{\pi^2    (T+1)}{3\alpha }\right ) + 2 \log j \right )  \frac{1}{2j N_{\textnormal{shots}} }} \\
	& \,\,  \leq \sqrt{ \left ( \log \left ( \frac{\pi^2    (T+1)}{3\alpha }\right ) + 2 \sqrt{j} \right )  \frac{1}{2j N_{\textnormal{shots}} }} \\
	& \,\, \leq \sqrt{  2\max \left (\log \left ( \frac{\pi^2    (T+1)}{3\alpha }\right ) \frac{1}{2j N_{\textnormal{shots}} }, 2 \sqrt{j} \frac{1}{2j N_{\textnormal{shots}} }\right ) } \\
	& \,\, \leq \frac{1}{2} \sin^2 \left ( \sqrt{\frac{r_t}{2}} \frac{1}{K} \frac{\pi}{2} \right ).
	\end{align*}
	From Lemma \ref{thm:p_to_theta},
	\begin{align*}
	& \,\,\left |\arcsin \sqrt{L_t}-\arcsin \sqrt{U_t} \right | \leq \arcsin \sqrt{|L_t-U_t | } \leq \arcsin \sqrt{2\delta_t}\leq \sqrt {\frac{r_t}{2}} \frac{1}{K} \frac{\pi}{2} \\
	\Rightarrow & \,\, |\check{\theta}_{t}^L-\check{\theta}_{t}^U| =\frac{1}{2m_t+1} \left |\arcsin \sqrt{L_t}-\arcsin \sqrt{U_t} \right | \leq \sqrt {\frac{r_t}{2}} \frac{1}{K} \frac{1}{2m_t+1} \frac{\pi}{2} \\
	\Rightarrow & \,\, \left |\min \left (\check{\theta}_{t}^L,\arcsin \sqrt{r_t/2} \right)-\min \left (\check{\theta}_{t}^U,\arcsin \sqrt{r_t/2} \right) \right | \leq \sqrt {\frac{r_t}{2}}\frac{1}{K} \frac{1}{2m_t+1} \frac{\pi}{2}.
	\end{align*}
	From Lemma \ref{thm:stretch},
	\begin{align*}
	 & \,\, 	|\theta_t^L -\theta_t^U| = \left |\arcsin \sqrt{\sin^2 \left ( \min \left (\check{\theta}_{t}^L,\arcsin \sqrt{r_{t}/2} \right) \right) /r_t}-\arcsin \sqrt{\sin^2 \left ( \min \left (\check{\theta}_{t}^U,\arcsin \sqrt{r_{t}/2} \right) \right) /r_t} \right | \\
	& \leq \sqrt {\frac{2}{r_t}}  \left |\min \left (\check{\theta}_{t}^L,\arcsin \sqrt{r_t/2} \right)-\min \left (\check{\theta}_{t}^U,\arcsin \sqrt{r_t/2} \right) \right | \leq \frac{1}{K} \frac{1}{2m_t+1} \frac{\pi}{2}.
	\end{align*}
	
	We now give a bound for $1/r_t$.
	 When $r_t<1$,
	\begin{align}
	& \,\, r_t  = \left . \sin ^2 \left (\frac{\hat{k}_{t}+1}{2m_{t}+1} \frac{\pi}{2} \right ) \right / \sin^2 (\theta_{t-1}^U) \nonumber \\ 
	\Rightarrow & \,\, \frac{1}{r_t} = \frac{\sin^2 (\theta_{t-1}^U)}{\sin ^2 \left (\frac{\hat{k}_{t}+1}{2m_{t}+1} \frac{\pi}{2} \right ) } \leq \frac{\sin ^2 \left (\frac{\hat{k}_{t}+2}{2m_{t}+1} \frac{\pi}{2} \right )}{\sin ^2 \left (\frac{\hat{k}_{t}+1}{2m_{t}+1} \frac{\pi}{2} \right ) } \leq \frac{\sin ^2 \left (\frac{\hat{k}_{t}+2}{2m_{t}+1} \frac{\pi}{2} \right )}{\sin ^2 \left (\frac{\hat{k}_{t}/2+1}{2m_{t}+1} \frac{\pi}{2} \right ) } = \left (2 \cos \left (\frac{\hat{k}_{t}/2+1}{2m_{t}+1} \frac{\pi}{2} \right ) \right )^2 \leq 4. \label{bound_r}
	\end{align}
	Therefore, 
	\begin{align*}
	N_t = O \left ( \log \left ( \frac{\pi^2 (T+1) }{3 \alpha} \right ) \right ).
	\end{align*}
	We now bound $N_{\textnormal{oracle}}$. Since $\hat{T}$ is the smallest number  such that $|(\sin \theta_{\hat{T}}^L)^2-(\sin \theta_{\hat{T}}^U)^2| \leq \epsilon$, $|(\sin \theta_{\hat{T}-1}^L)^2-(\sin \theta_{\hat{T}-1}^U)^2| \geq \epsilon$. And recall $| \theta_{\hat{T}-1}^L-\theta_{\hat{T}-1}^U | \leq  \frac{1}{2m_{\hat{T}}+1} \frac{\pi}{2}$, which gives
	\begin{align*}
	\epsilon \leq |(\sin \theta_{\hat{T}-1}^L)^2-(\sin \theta_{\hat{T}-1}^U)^2| \leq 2 | \theta_{\hat{T}-1}^L-\theta_{\hat{T}-1}^U | \leq  \frac{ \pi}{2m_{\hat{T}}+1} \Rightarrow 2m_{\hat{T}}+1 \leq \frac{\pi}{\epsilon}.
	\end{align*}
	And since $2m_{t+1}+1 \geq K(2m_t+1)$, a simple induction argument shows 
	\begin{align*}
	2m_t+1 \leq \frac{1}{K^{\hat{T}-t}} \frac{\pi}{\epsilon}, \,\, t=0,...,\hat{T}.
	\end{align*}
	Finally,
	\begin{align*}
	N_{\textnormal{oracle}} & =  \sum_{t=0}^{\hat{T}} N_t m_t \leq O \left ( \log \left ( \frac{\pi^2 (T+1) }{3 \alpha} \right ) \right ) \sum_{t=0}^{\hat{T}}  (2m_t+1) \leq O \left ( \log \left ( \frac{\pi^2 (T+1) }{3 \alpha} \right ) \right )   \frac{\pi}{ \epsilon} \sum_{t=0}^{\hat{T}}  \frac{1}{K^t}\\
	& \leq O \left ( \log \left ( \frac{\pi^2 (T+1) }{3 \alpha} \right ) \right )  \frac{\pi}{ \epsilon} \sum_{t=0}^{\infty} \frac{1}{K^t}  = O \left ( \log \left ( \frac{\pi^2 (T+1) }{3 \alpha} \right ) \frac{1}{ \epsilon} \right ).
	\end{align*}
\end{proof}

\section{Numerical Experiments}\label{sec:numerical}
We compare through numerical experiments the proposed adaptive algorithm to two other algorithms, the maximum likelihood amplitude estimation (MLAE) and the iterative quantum amplitude estimation (IQAE). We use the MLAE and IQAE algorithms provided in Qiskit, an open source software development kit for quantum computing. For comparison purposes, we also use quantum simulators and circuits in Qiskit when implementing  the adaptive algorithm.  In all algorithms $X_t$ is sampled from a binomial distribution with probability $\sin^2 ( (2m_t+1)\theta)$ (with proper adjustment on $\theta$ in the adaptive algorithm). Therefore, the  time costs reported below reflect  the computation complexity of the classical part of the algorithms.  We choose the  confidence level as 95\% in all algorithms. For IQAE and the adaptive algorithm, we set the target precision as $\{10^{-3},10^{-4},...,10^{-10}\}$. Instead of specifying the target precision, MLAE requires an input of the number of iterations $T$ and chooses $\{m_t\}_{t=0}^{T}$  as $m_0=0,m_1=2^0,...,m_T=2^{T-1}.$   We use $T=8,10,12$ and $14$. Furthermore, we choose $K=3$ in the adaptive algorithm. 

We compare the algorithms in three scenarios. We choose $N_{\textnormal{shots}}=100$ in all algorithms in the first two scenarios. In the first scenario, we sample 100 values of $p$ uniformly from 0 to 0.5.  Each point in Panel (a) (b) and (c) of Figure \ref{fig:sim} is an average from the 100 experiments. The findings are summarized as follows. Firstly, from Panel (a) the adaptive algorithm requires slightly more oracle queries than MLAE and IQAE to achieve the same level of precision. MLAE uses the likelihood-ratio  method to construct the confidence interval, which lacks for rigorous justification.  When implementing IQAE, we chose the Clopper-Pearson method, which was not justified completely analytically (Supplementary information to \cite{grinko2021iterative}, Theorem 1). We attempted to conduct experiments using IQAE with the Chernoff-Hoeffding method, which gives more conservative but theoretically justifiable intervals and is in line with the choice in the adaptive method. But the Qiskit version of IQAE with the Chernoff-Hoeffding method using $N_{\textnormal{shots}}=100$ could not produce outcomes within a reasonable time. We will increase $N_{\textnormal{shots}}$ in the third scenario for comparison. 

Secondly, the time costs of the classical part of the adaptive algorithm are substantially less than MLAE and IQAE from Panel (b). By \cite{suzuki2020amplitude}, the computational complexity\footnote{Here we treat $\alpha$ as a constant.} of the classical part of MLAE  is $O(1/\epsilon \log (1/\epsilon))$, which is in line with Panel (b). Due to its high time costs, we will not compare  MLAE in the following scenarios. By contrast, the computational complexity of the classical part of the adaptive algorithm is $O(\log (1/\epsilon) \log(\log(1/\epsilon)))$ because $T$ is $O(\log (1/\epsilon))$ and the runtime in each step $t$ is proportional to the iterations in the \tb{repeat} loop, which is $O(\log(\log(1/\epsilon))$ by  Theorem \ref{thm:main}. The time costs of IQAE show a similar pattern but the average time cost is approximately 40 times of the adaptive algorithm. 

Thirdly, we report the average value of $r_t$ and the worst-case value, i.e., the smallest value across all steps. We gave a theoretical lower bound of $r_t$ as $1/4$ in \eqref{bound_r}. From Panel (c), the average value of $r_t$ is close to 1 and increases with the precision, which implies that the estimation loses very little efficiency due to the adjustment on average. 
  The worst-case value is between 0.6 and 0.7. 

Finally (not shown in the figure), 100\% of the intervals by the adaptive algorithm and IQAE contain the true values of $p$ in all experiments, and 98.5\% of the intervals by MLAE contain the true values of $p$. 

In the second scenario, we compare the adaptive algorithm and IQAE at a specific value $p=0.25$, which corresponds to $\theta=\pi/6$, a boundary between periods of $\sin^2(K^t \theta)$ for $K=3$. Each point in Panel (d) (e) and (f) of Figure \ref{fig:sim} is an average from the 100 experiments. As before, 100\% of the intervals by the adaptive algorithm and IQAE contain the true values of $p$ in all experiments. The gaps between the  numbers of  oracle queries  of the two methods becomes slightly larger.  But a more notable pattern is the rapid growth of the runtime of IQAE when $\epsilon$ is small. The bottleneck of IQAE is the sub-routine FINDNEXTK, which performs the following task (using the notation in this paper): recall that $m_{t+1}$ is the largest integer such that  $|\theta^L_t-\theta^U_t | \leq  \frac{1}{2m_{t+1}+1}\frac{\pi}{2}$. The sub-routine starts from  $m_{t+1}$ and gradually decreases this number until reach $\tilde{m}_{t+1}$ such that  $[\theta^L_t,\theta^U_t]$ is fully contained in a single length-$\frac{1}{2\tilde{m}_{t+1}+1}\frac{\pi}{2}$ period of $\sin^2((2\tilde{m}_{t+1}+1)\theta)$. In the worst-case scenario, the runtime of FINDNEXTK can be proportional to $m_{t+1}$ and eventually be $O(1/\epsilon)$, which is demonstrated in Panel (e).  Finally, the adjustment factor $r_t$, especially in the worst-case scenario, is smaller than the corresponding value in the previous simulation. That is because  $[\theta^L_t,\theta^U_t]$ is more likely to overlap with two periods since $\theta$ is at the boundary. 

\begin{figure}[!htp]
	\twoImages{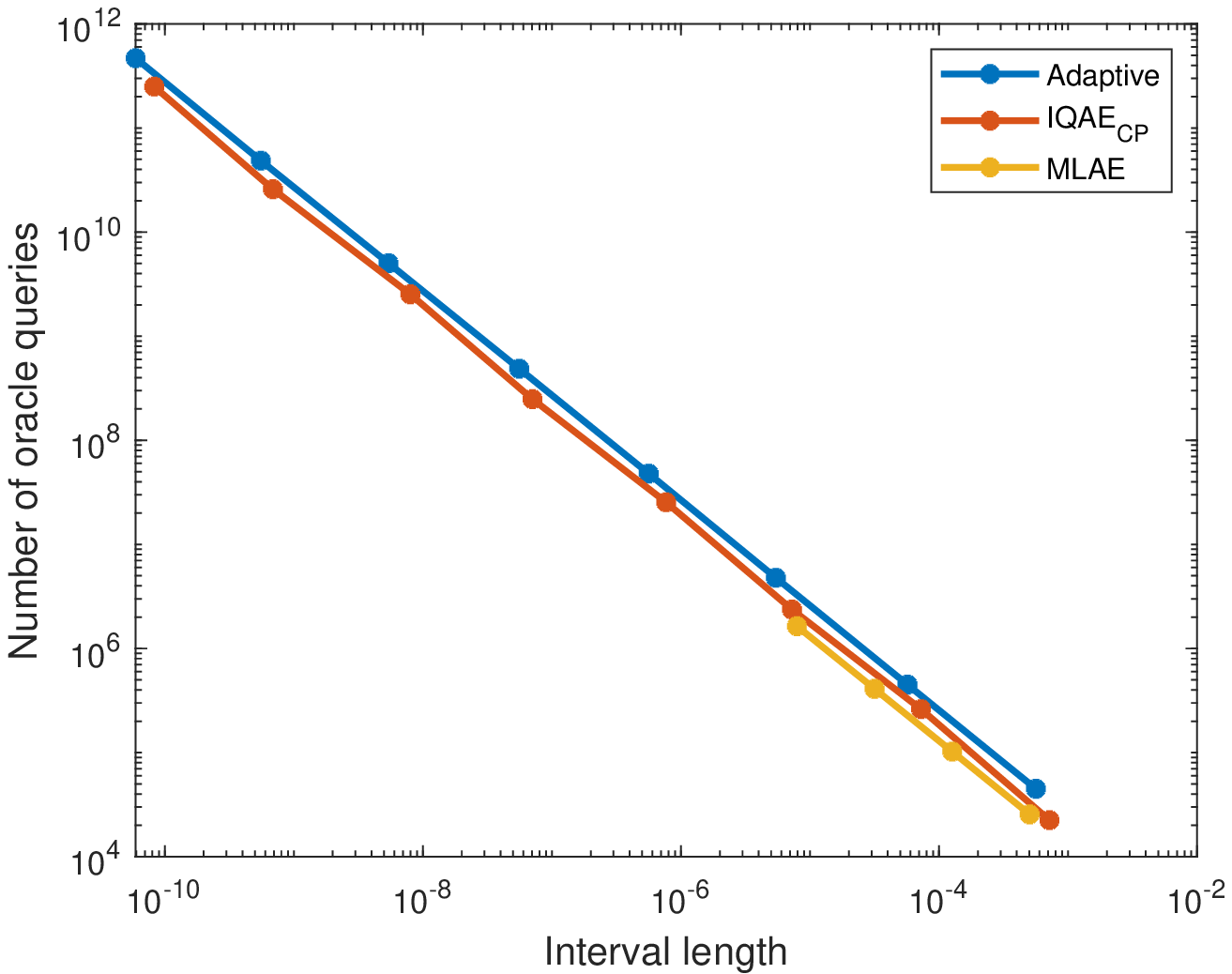}{8cm}{(a) Log-log plot of $N_\textnormal{oracles}$ for $p\sim U(0,0.5)$}{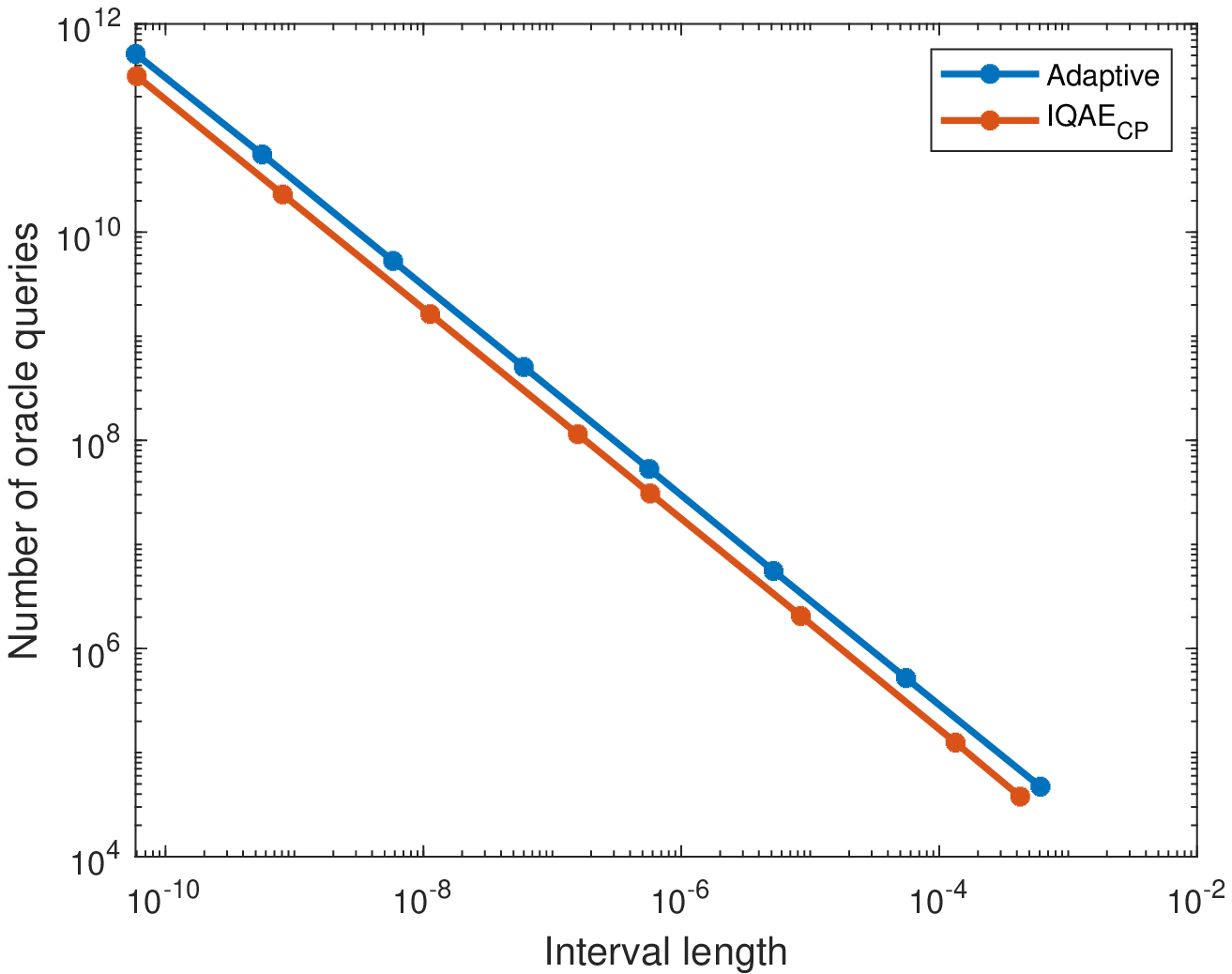}{8cm}{(d) Log-log plot of $N_\textnormal{oracles}$ for $p=0.25$}
	\twoImages{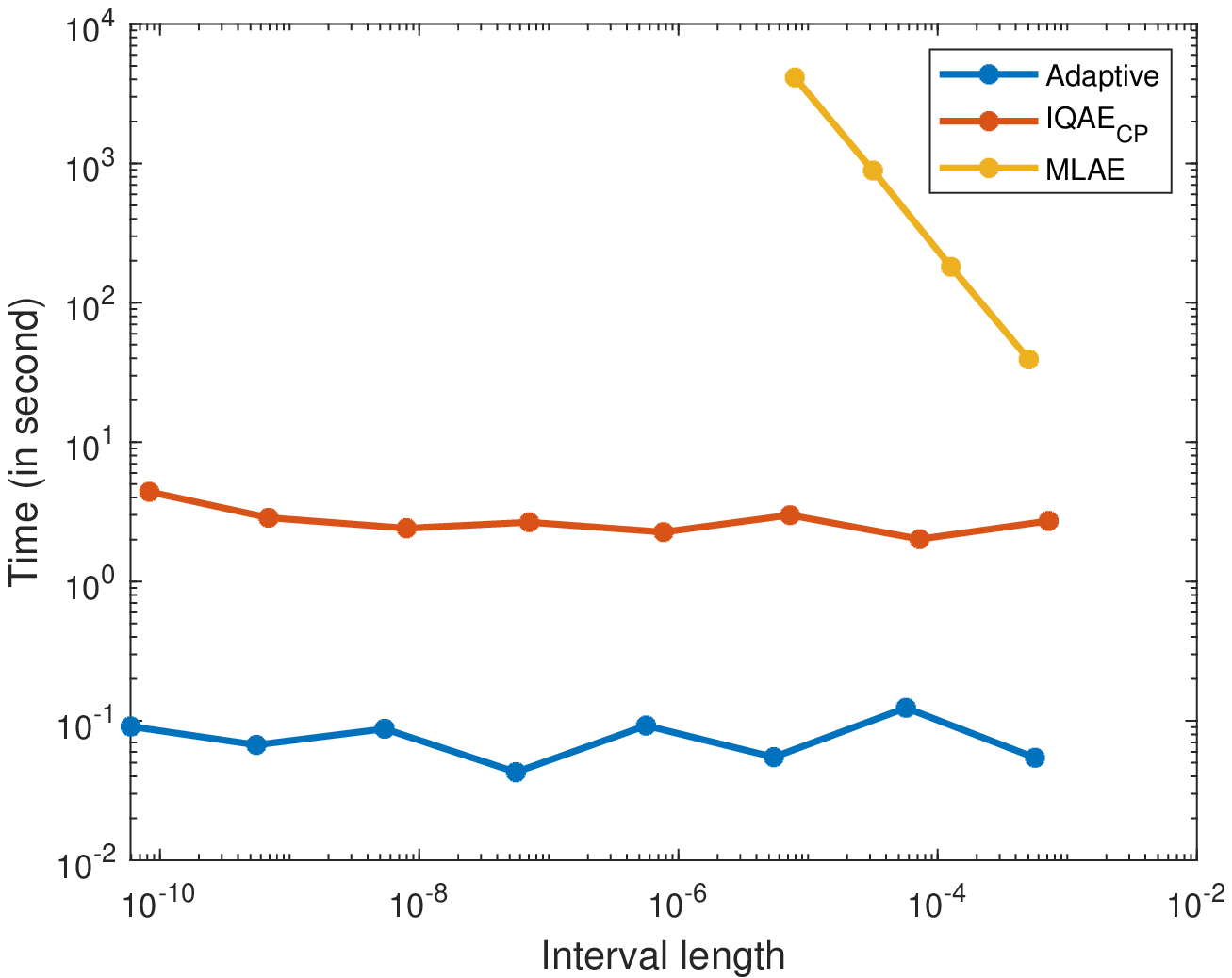}{8cm}{(b) Log-log plot of time costs for $p\sim U(0,0.5)$}{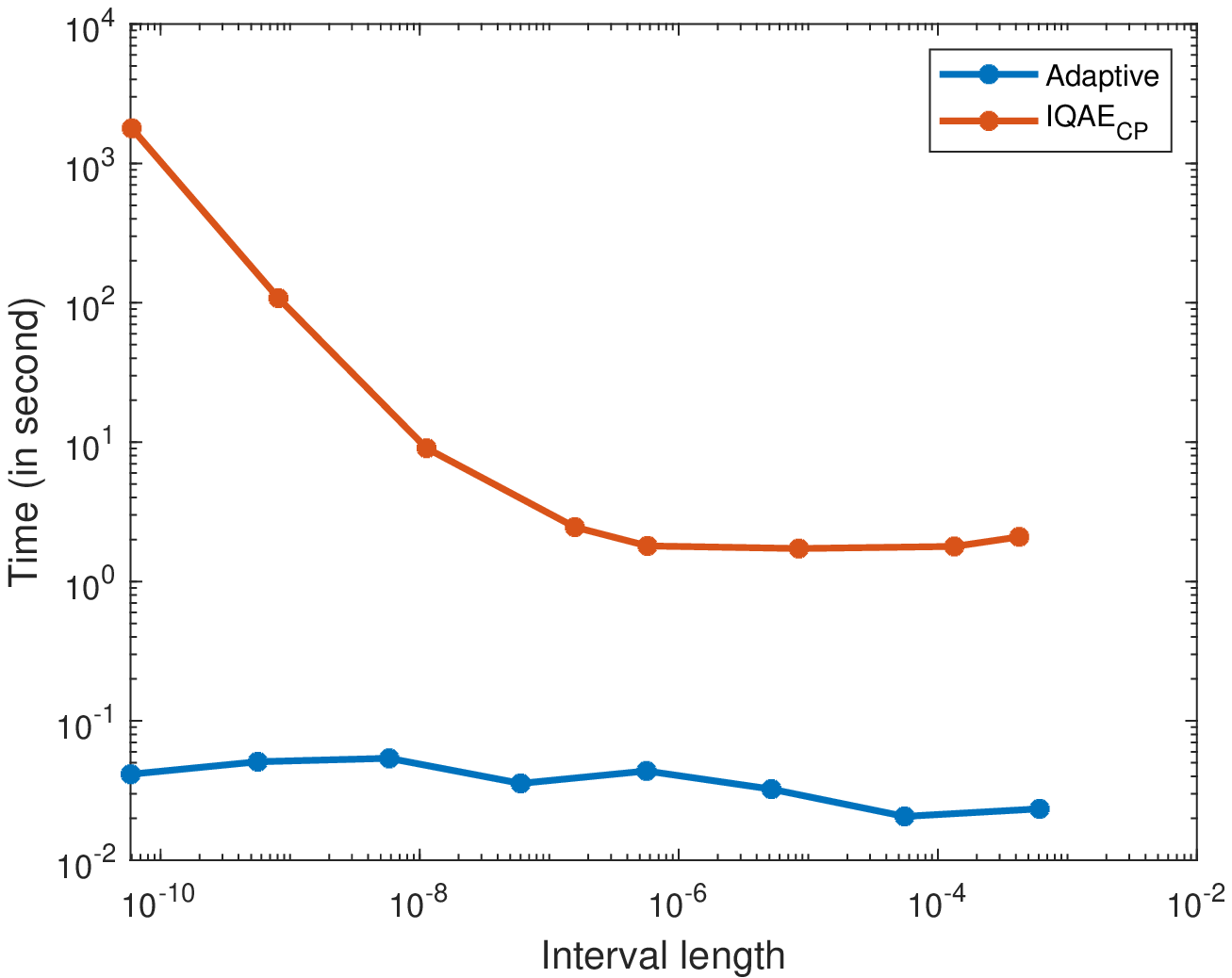}{8cm}{(e) Log-log plot of time costs for $p=0.25$}
	\twoImages{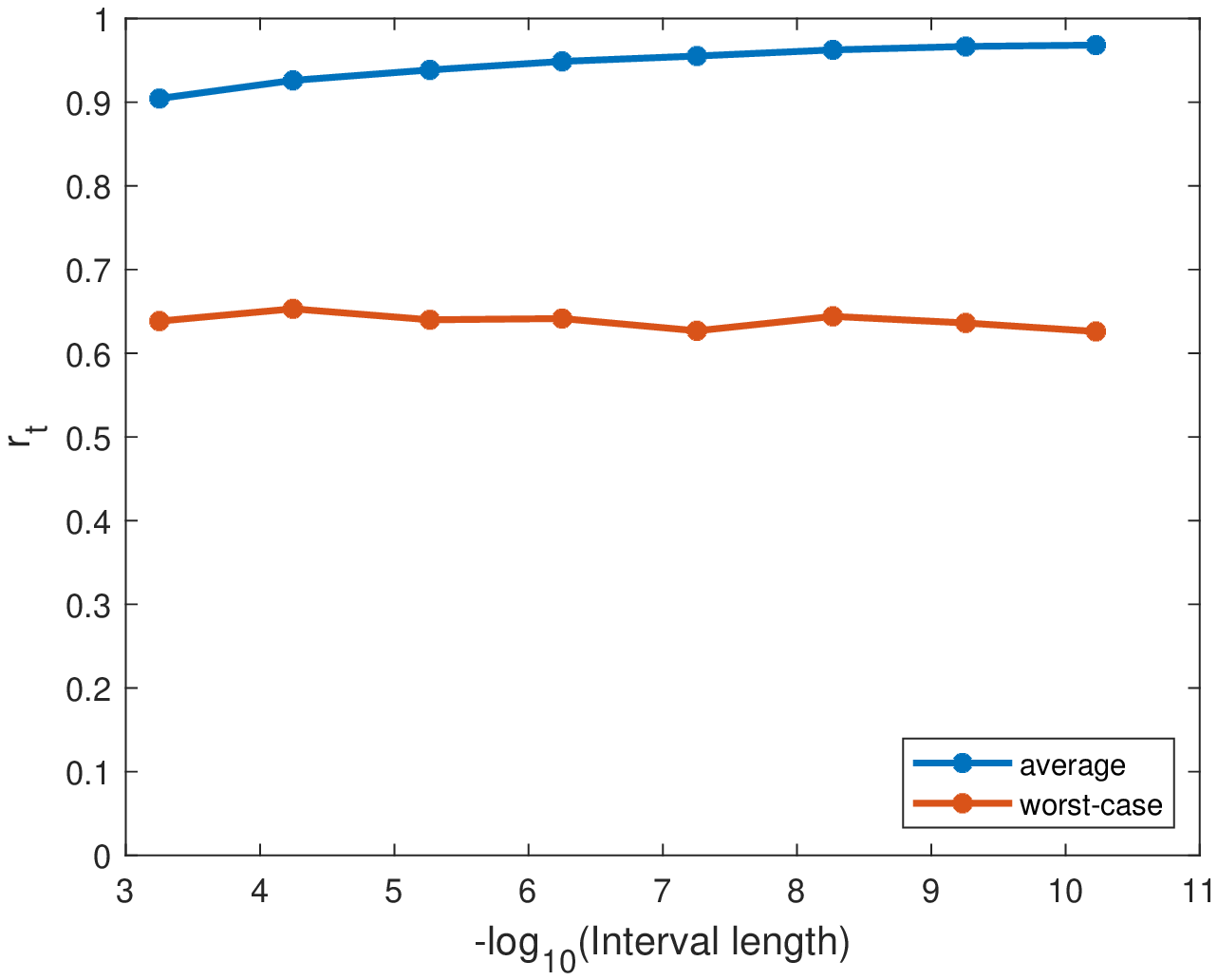}{8cm}{(c) Adjustment factor $r_t$ for $p\sim U(0,0.5)$}{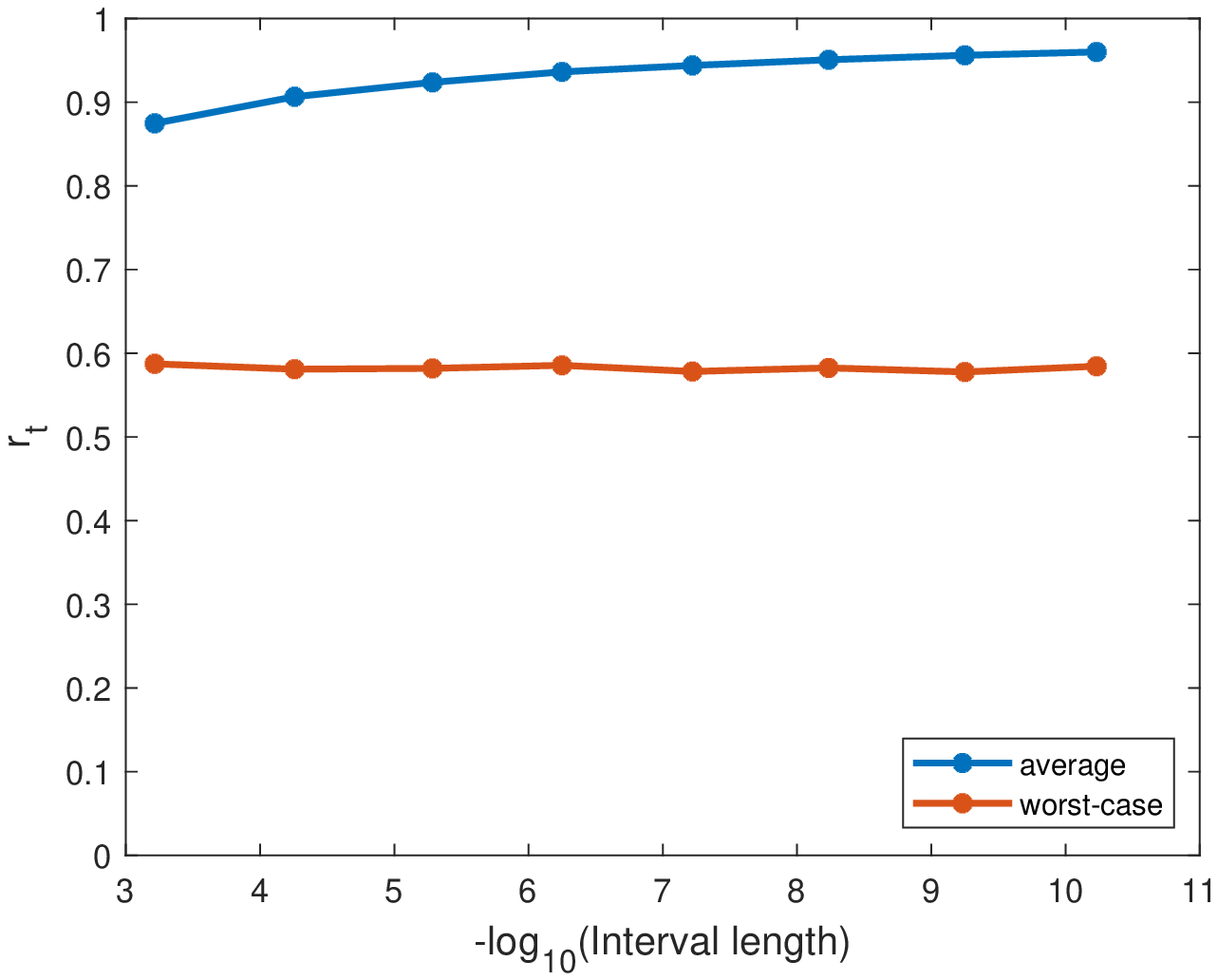}{8cm}{(f)  Adjustment factor $r_t$  for $p=0.25$}

	\caption{Comparison of MLAE, IQAE, and the adaptive algorithm using $N_{\textnormal{shots}}=100$. Confidence level $=95\%$.  Each point in Panel (a) (b) and (c) is an average of the experimental results for 100 values of $p$ sampled from $U(0,0.5)$. Each point in Panel (d) (e) and (f) is an average of 100 experimental results for $p=0.25$. The average $r_t$ refers to the average value of $r_t$ over $t=0,...,T$ and the worst-case $r_t$ refers to the minimum value of $r_t$ over $t=0,...,T$. }
	
	\label{fig:sim}
\end{figure}

In the third scenario, we compare the adaptive algorithm and IQAE with the Clopper-Pearson method ($\textnormal{IQAE}_{\textnormal{CP}}$) and the Chernoff-Hoeffding method ($\textnormal{IQAE}_{\textnormal{CH}}$).  We  use $N_{\textnormal{shots}}=800$ in all three methods for a fair comparison because IQAE with the Chernoff-Hoeffding method using a smaller $N_{\textnormal{shots}}$ sometimes could not return an output. The rest of the setup is identical to the first scenario. As aforementioned, the Chernoff-Hoeffding method gives a more conservative confidence interval but with a theoretical guarantee. The same method is also used in the adaptive algorithm. From Figure \ref{fig:sim2}, the number of oracle queries by the adaptive algorithm  is between $\textnormal{IQAE}_{\textnormal{CP}}$ and $\textnormal{IQAE}_{\textnormal{CH}}$ under the same level of precision. This suggests that the adaptive algorithm uses a slightly smaller number of queries than IQAE when using the same method for constructing confidence intervals of $\sin^2 ((2m_t+1)\theta)$. Moreover, the classical part of the adaptive algorithm has substantially lower computational complexity than IQAE as in the previous scenarios. Finally, 100\% of the intervals by the adaptive algorithm and both versions of IQAE contain the true values of $p$ in all experiments.

\begin{figure}[!htp]
	\twoImages{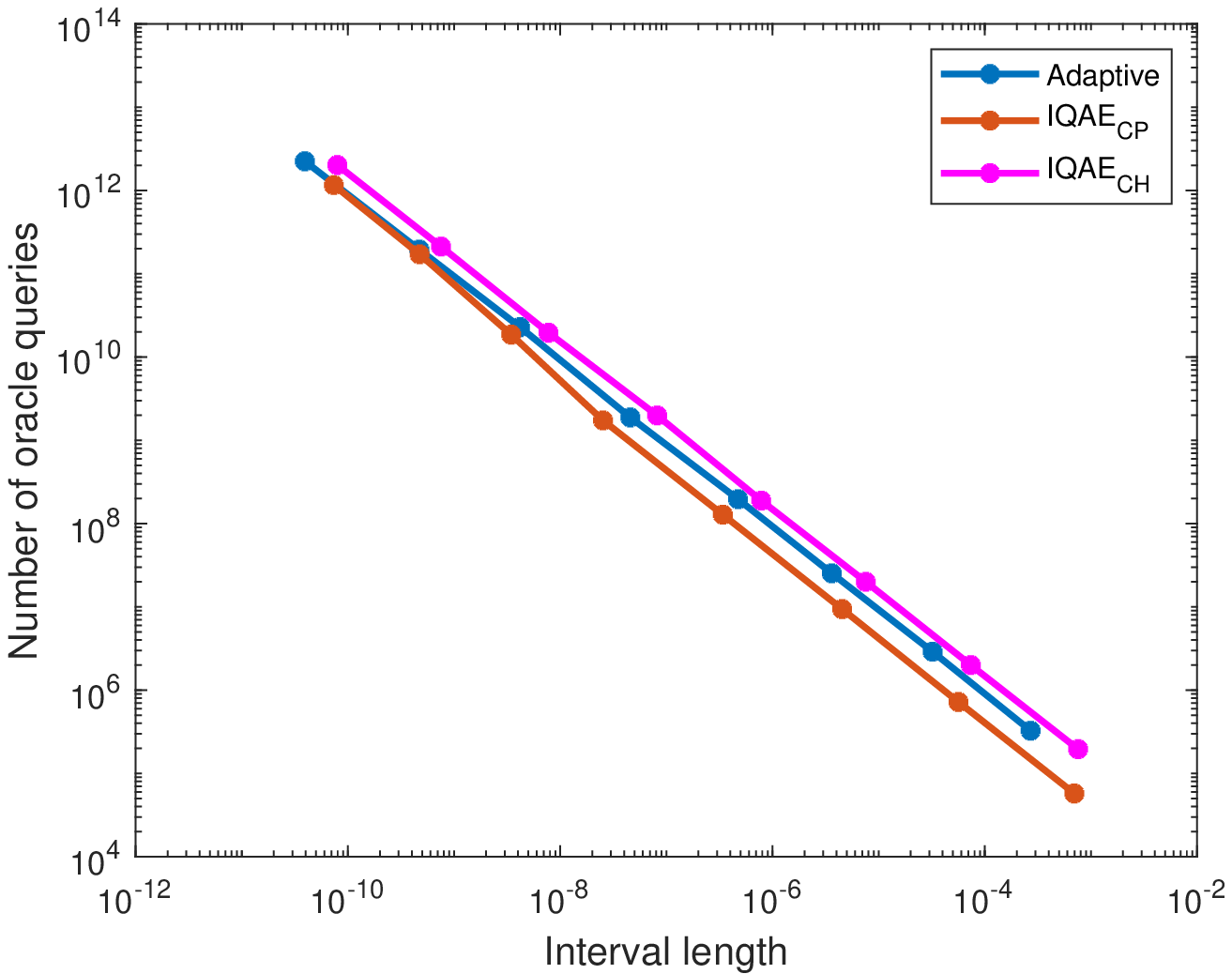}{8cm}{(a) Log-log plot of $N_\textnormal{oracles}$ for $p\sim U(0,0.5)$}{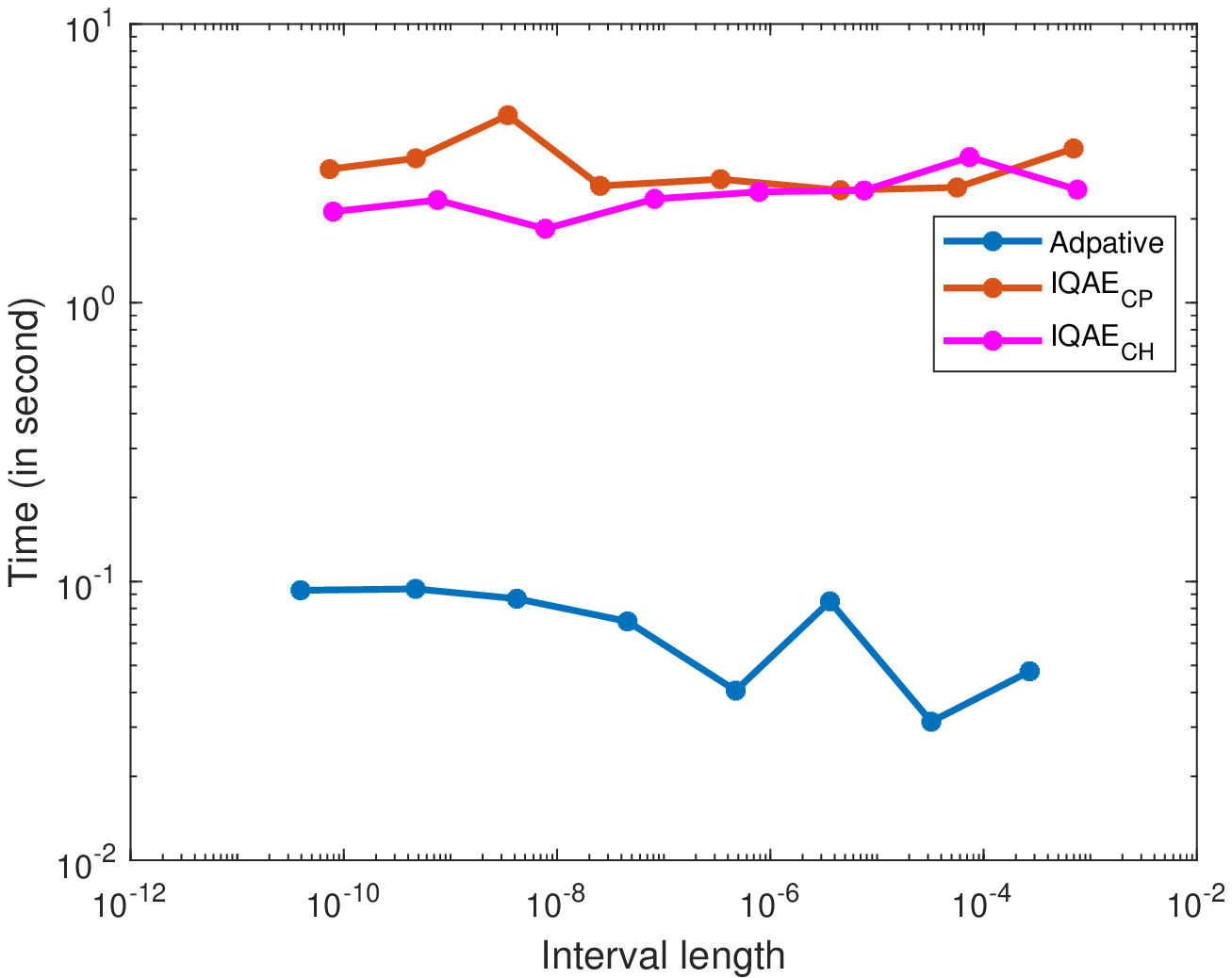}{8cm}{(d) Log-log plot of time costs for $p\sim U(0,0.5)$}
	\caption{Comparison of MLAE, IQAE with the Clopper-Pearson method ($\textnormal{IQAE}_{\textnormal{CP}}$), and IQAE with the Chernoff-Hoeffding method ($\textnormal{IQAE}_{\textnormal{CH}}$) using $N_{\textnormal{shots}}=800$. Confidence level $=95\%$.  Each point  is an average of the experimental results for 100 values of $p$ sampled from $U(0,0.5)$.  }
		\label{fig:sim2}

\end{figure}

\section{Conclusion}\label{sec:conclusion}
We proposed a new  Grover-based amplitude estimation algorithm. The number of oracle queries achieves $O(1/\epsilon)$ and the computational complexity of the classical part achieves $O(\log(1/\epsilon))$, both up to a  double-logarithmic factor. The key ingredient of the algorithm is an adjustment factor $r_t$ such that the confidence interval for $\theta_t= \arcsin \sqrt{r_t p}$ is fully contained in a single period as long as the length of the original interval for $\theta$ does not exceed  the length of the period. With this adjustment, the algorithm does not need to search for the appropriate number of Grover iterations in each step, which can be time-consuming, and both the number of total steps and the number of measurements are easy to bound analytically. 

The theoretical result in this paper (Theorem \ref{thm:main}) is a non-asymptotic result in nature. In fact, such a non-asymptotic result is easier to formulate than an asymptotic result in this scenario because the number of measurements in a single step does not go to infinity. Therefore, a non-asymptotic bound such as Hoeffding's inequality can be naturally applied. But such a non-asymptotic bound can be loose. One may therefore be interested in  the asymptotic distribution of $N_{\textnormal{oracles}} (\hat{p}-p)$ where $\hat{p}$ is an estimator of $p$ such as the maximum likelihood estimator. A related but simpler problem is to derive the asymptotic variance of the estimator. 


\section*{Appendix}
   \begin{lemma}\label{thm:shorten}
	For $0\leq \theta_1,\theta_2 \leq \pi/2$, $0\leq r \leq 1$, 
	\begin{align*}
	|\arcsin \sqrt{r(\sin \theta_1)^2}-\arcsin \sqrt{r(\sin \theta_2)^2} |\leq  |\theta_1-\theta_2|.
	\end{align*}
\end{lemma}
\begin{proof}
	Without loss of generality, assume $\theta_1 > \theta_2$. Let 
	\begin{align*}
	g(r)=\arcsin \sqrt{r(\sin \theta_1)^2}-\arcsin \sqrt{r(\sin \theta_2)^2}.
	\end{align*}
	Note that $g(1)=\theta_1-\theta_2$. Then we only need to prove $g(r)$ is a non-decreasing function. In fact,
	\begin{align*}
	& g'(r)= \frac{1}{\sqrt{1-r(\sin \theta_1)^2}} \sin \theta_1  \frac{1}{2\sqrt{r}}  -\frac{1}{\sqrt{1-r(\sin \theta_2)^2}} \sin \theta_2 \frac{1}{2\sqrt{r}} \geq 0 \\
	\Leftrightarrow & \,\, \frac{\sin \theta_1}{\sqrt{1-r(\sin \theta_1)^2}} \geq  \frac{\sin \theta_2}{\sqrt{1-r(\sin \theta_2)^2}} \\
	\Leftrightarrow & \,\, (\sin \theta_1)^2 (1-r(\sin \theta_2)^2) \geq (\sin \theta_2)^2 (1-r(\sin \theta_1)^2) \\
	\Leftrightarrow & \,\, (\sin \theta_1)^2 \geq  (\sin \theta_2)^2.
	\end{align*}
\end{proof}
 \begin{lemma}\label{thm:stretch}
	For $0\leq \theta_1,\theta_2 \leq \pi/2$, $s\geq 1$, satisfying $s(\sin \theta_1)^2\leq 1/2$ and $s(\sin \theta_2)^2\leq 1/2$,
	\begin{align*}
	|\arcsin \sqrt{s(\sin \theta_1)^2}-\arcsin \sqrt{s(\sin \theta_2)^2} |\leq  \sqrt{2s} |\theta_1-\theta_2|.
	\end{align*}
\end{lemma}
\begin{proof}
	\begin{align*}
	\left | \arcsin \sqrt{s(\sin \theta_1)^2}-\arcsin \sqrt{s(\sin \theta_2)^2} \right |  = \left | \frac{\sqrt{s} \cos(\tilde{\theta})}{\sqrt{1-s(\sin \tilde{\theta})^2}} (\theta_1-\theta_2) \right | \leq \sqrt{2s} |\theta_1-\theta_2|.
	\end{align*} 
\end{proof}
 \begin{lemma}\label{thm:p_to_theta}
	For $0\leq p_1,p_2\leq 1$, 
	\begin{align*}
	|\arcsin \sqrt{p_1}-\arcsin \sqrt{p_2}| \leq \arcsin{\sqrt{|p_1-p_2|}}.
	\end{align*}	
\end{lemma}
\begin{proof}
	Let $x=p_2$ and $\delta=p_1-p_2$. Without loss of generality, assume $0<\delta<1$. Consider the function $f(x)=\arcsin \sqrt{x+\delta} -\arcsin \sqrt{x}$. We only need to prove
	\begin{align*}
	\max_{x \in [0,1-\delta]} f(x) =f(0).
	\end{align*} 
	Notice 
	\begin{align*}
	f'(x)=\frac{1}{2\sqrt{x+\delta} \sqrt{1-(x+\delta)}} -\frac{1}{2\sqrt{x} \sqrt{1-x}}.
	\end{align*}
	The only stationary point of $f(x)$ on $[0,1-\delta]$ is $x=\frac{1}{2} (1-\delta)$. Moreover, $\lim_{x\rightarrow 0} f'(x)=-\infty$ and $\lim_{x\rightarrow 1-\delta} f'(x)=\infty$. By the intermediate value theorem, $f'(x)<0$ for $x \in \left (0,\frac{1}{2} (1-\delta) \right )$ and $f'(x)>0$ for $x \in \left (\frac{1}{2} (1-\delta),1-\delta \right )$. By the mean value theorem, for all $x \in \left (0, \frac{1}{2} (1-\delta) \right ]$, $f(x)-f(0)=x f'(\tilde{x})<0$ where $\tilde{x}\in (0,x)$. Similarly, for $x \in \left [  \frac{1}{2} (1-\delta),1-\delta \right )$, $f(1-\delta)-f(x)>0$. Therefore, the maximum value of $f(x)$ can only be achieved at the two endpoints. In fact, $f(0)=f(1-\delta)=\arcsin \sqrt{\delta}$.
\end{proof}

 \newcommand{\noop}[1]{}


\begin{thebibliography}{}

\bibitem[Aaronson and Rall, 2020]{aaronson2020quantum}
Aaronson, S. and Rall, P. (2020).
\newblock Quantum approximate counting, simplified.
\newblock In {\em Symposium on Simplicity in Algorithms}, pages 24--32. SIAM.

\bibitem[Ambainis, 2004]{ambainis2004quantum}
Ambainis, A. (2004).
\newblock Quantum search algorithms.
\newblock {\em ACM SIGACT News}, 35(2):22--35.

\bibitem[Artiles et~al., 2005]{gill_jrssb}
Artiles, L.~M., Gill, R.~D., and Guţă, M.~I. (2005).
\newblock An invitation to quantum tomography.
\newblock {\em Journal of the Royal Statistical Society. Series B (Statistical
  Methodology)}, 67(1):109--134.

\bibitem[Brassard et~al., 2002]{brassard2002quantum}
Brassard, G., Hoyer, P., Mosca, M., and Tapp, A. (2002).
\newblock Quantum amplitude amplification and estimation.
\newblock {\em Contemporary Mathematics}, 305:53--74.

\bibitem[Cao et~al., 2019]{cao2019quantum}
Cao, Y., Romero, J., Olson, J.~P., Degroote, M., Johnson, P.~D., Kieferov{\'a},
  M., Kivlichan, I.~D., Menke, T., Peropadre, B., Sawaya, N.~P., et~al. (2019).
\newblock Quantum chemistry in the age of quantum computing.
\newblock {\em Chemical reviews}, 119(19):10856--10915.

\bibitem[Durr and Hoyer, 1996]{durr1996quantum}
Durr, C. and Hoyer, P. (1996).
\newblock A quantum algorithm for finding the minimum.
\newblock {\em arXiv preprint quant-ph/9607014}.

\bibitem[Egger et~al., 2020]{egger2020credit}
Egger, D.~J., Guti{\'e}rrez, R.~G., Mestre, J.~C., and Woerner, S. (2020).
\newblock Credit risk analysis using quantum computers.
\newblock {\em IEEE Transactions on Computers}, 70(12):2136--2145.

\bibitem[Gill, 2008]{gill2008conciliation}
Gill, R.~D. (2008).
\newblock Conciliation of bayes and pointwise quantum state estimation.
\newblock In {\em Quantum Stochastics and Information: Statistics, Filtering
  and Control}, pages 239--261. World Scientific.

\bibitem[Gill and Gu{\c{t}}{\u{a}}, 2013]{gill2013asymptotic}
Gill, R.~D. and Gu{\c{t}}{\u{a}}, M.~I. (2013).
\newblock On asymptotic quantum statistical inference.
\newblock In {\em From Probability to Statistics and Back: High-Dimensional
  Models and Processes--A Festschrift in Honor of Jon A. Wellner}, pages
  105--127. Institute of Mathematical Statistics.

\bibitem[Grinko et~al., 2021]{grinko2021iterative}
Grinko, D., Gacon, J., Zoufal, C., and Woerner, S. (2021).
\newblock Iterative quantum amplitude estimation.
\newblock {\em npj Quantum Information}, 7(1):1--6.

\bibitem[Grover, 1996]{grover1996fast}
Grover, L.~K. (1996).
\newblock A fast quantum mechanical algorithm for database search.
\newblock In {\em Proceedings of the twenty-eighth annual ACM symposium on
  Theory of computing}, pages 212--219.

\bibitem[Herman et~al., 2022]{herman2022survey}
Herman, D., Googin, C., Liu, X., Galda, A., Safro, I., Sun, Y., Pistoia, M.,
  and Alexeev, Y. (2022).
\newblock A survey of quantum computing for finance.
\newblock {\em arXiv preprint arXiv:2201.02773}.

\bibitem[Hoeffding, 1963]{Hoeffding63}
Hoeffding, W. (1963).
\newblock Probability inequalities for sums of bounded random variables.
\newblock {\em Journal of the American Statistical Association},
  58(301):13--30.

\bibitem[Hong et~al., 2014]{hong2014monte}
Hong, L.~J., Hu, Z., and Liu, G. (2014).
\newblock Monte carlo methods for value-at-risk and conditional value-at-risk:
  a review.
\newblock {\em ACM Transactions on Modeling and Computer Simulation (TOMACS)},
  24(4):1--37.

\bibitem[Hu and Wang, 2020]{hu2020quantum}
Hu, J. and Wang, Y. (2020).
\newblock Quantum annealing via path-integral monte carlo with data
  augmentation.
\newblock {\em Journal of Computational and Graphical Statistics},
  30(2):284--296.

\bibitem[Kassal et~al., 2008]{kassal2008polynomial}
Kassal, I., Jordan, S.~P., Love, P.~J., Mohseni, M., and Aspuru-Guzik, A.
  (2008).
\newblock Polynomial-time quantum algorithm for the simulation of chemical
  dynamics.
\newblock {\em Proceedings of the National Academy of Sciences},
  105(48):18681--18686.

\bibitem[Kitaev, 1995]{kitaev1995quantum}
Kitaev, A.~Y. (1995).
\newblock Quantum measurements and the abelian stabilizer problem.
\newblock {\em arXiv preprint quant-ph/9511026}.

\bibitem[Knill et~al., 2007]{knill2007optimal}
Knill, E., Ortiz, G., and Somma, R.~D. (2007).
\newblock Optimal quantum measurements of expectation values of observables.
\newblock {\em Physical Review A}, 75(1):012328.

\bibitem[Kochenberger et~al., 2014]{kochenberger2014unconstrained}
Kochenberger, G., Hao, J.-K., Glover, F., Lewis, M., L{\"u}, Z., Wang, H., and
  Wang, Y. (2014).
\newblock The unconstrained binary quadratic programming problem: a survey.
\newblock {\em Journal of combinatorial optimization}, 28(1):58--81.

\bibitem[Montanaro, 2015]{montanaro2015quantum}
Montanaro, A. (2015).
\newblock Quantum speedup of monte carlo methods.
\newblock {\em Proceedings of the Royal Society A: Mathematical, Physical and
  Engineering Sciences}, 471(2181):20150301.

\bibitem[Nakaji, 2020]{nakaji2020faster}
Nakaji, K. (2020).
\newblock Faster amplitude estimation.
\newblock {\em arXiv preprint arXiv:2003.02417}.

\bibitem[Nielsen and Chuang, 2011]{nielsen}
Nielsen, M.~A. and Chuang, I. (2011).
\newblock {\em Quantum Computation and Quantum Information: 10th Anniversary
  Edition}.
\newblock Cambridge University Press.

\bibitem[Pomerance, 1996]{pomerance1996tale}
Pomerance, C. (1996).
\newblock A tale of two sieves.
\newblock In {\em Notices Amer. Math. Soc}. Citeseer.

\bibitem[Ramezani et~al., 2020]{ramezani2020machine}
Ramezani, S.~B., Sommers, A., Manchukonda, H.~K., Rahimi, S., and Amirlatifi,
  A. (2020).
\newblock Machine learning algorithms in quantum computing: A survey.
\newblock In {\em 2020 international joint conference on neural networks
  (IJCNN)}, pages 1--8. IEEE.

\bibitem[Rebentrost et~al., 2018]{rebentrost2018quantum}
Rebentrost, P., Gupt, B., and Bromley, T.~R. (2018).
\newblock Quantum computational finance: Monte carlo pricing of financial
  derivatives.
\newblock {\em Physical Review A}, 98(2):022321.

\bibitem[Shor, 1994]{shor1994algorithms}
Shor, P.~W. (1994).
\newblock Algorithms for quantum computation: discrete logarithms and
  factoring.
\newblock In {\em Proceedings 35th annual symposium on foundations of computer
  science}, pages 124--134. Ieee.

\bibitem[Sun et~al., 2014]{sun2014quantum}
Sun, G., Su, S., and Xu, M. (2014).
\newblock Quantum algorithm for polynomial root finding problem.
\newblock In {\em 2014 Tenth International Conference on Computational
  Intelligence and Security}, pages 469--473. IEEE.

\bibitem[Suzuki et~al., 2020]{suzuki2020amplitude}
Suzuki, Y., Uno, S., Raymond, R., Tanaka, T., Onodera, T., and Yamamoto, N.
  (2020).
\newblock Amplitude estimation without phase estimation.
\newblock {\em Quantum Information Processing}, 19(2):1--17.

\bibitem[Wang, 2022]{Wang2022When}
Wang, Y. (2022).
\newblock When quantum computation meets data science: Making data science
  quantum.
\newblock {\em Harvard Data Science Review}, 4(1).
\newblock https://hdsr.mitpress.mit.edu/pub/kpn45eyx.

\bibitem[Wang and Liu, 2022]{wang2022quantum}
Wang, Y. and Liu, H. (2022).
\newblock Quantum computing in a statistical context.
\newblock {\em Annual Review of Statistics and Its Application}, 9.

\bibitem[Wang et~al., 2016]{wang2016quantum}
Wang, Y., Wu, S., and Zou, J. (2016).
\newblock Quantum annealing with markov chain monte carlo simulations and
  d-wave quantum computers.
\newblock {\em Statistical Science}, pages 362--398.

\bibitem[Wie, 2019]{wie2019simpler}
Wie, C.-R. (2019).
\newblock Simpler quantum counting.
\newblock {\em Quantum Information \& Computation}, 16(11-12):967–983.

\bibitem[Wiebe et~al., 2015]{wiebe2015quantum}
Wiebe, N., Kapoor, A., and Svore, K.~M. (2015).
\newblock Quantum algorithms for nearest-neighbor methods for supervised and
  unsupervised learning.
\newblock {\em Quantum Information \& Computation}, 15(3-4):316--356.

\bibitem[Wiebe et~al., 2016]{wiebe2016quantum}
Wiebe, N., Kapoor, A., and Svore, K.~M. (2016).
\newblock Quantum deep learning.
\newblock {\em Quantum Information \& Computation}, 16(7-8):541--587.

\bibitem[Woerner and Egger, 2019]{woerner2019quantum}
Woerner, S. and Egger, D.~J. (2019).
\newblock Quantum risk analysis.
\newblock {\em npj Quantum Information}, 5(1):1--8.

\bibitem[Zhong et~al., 2021]{zhong2021best}
Zhong, W., Ke, Y., Wang, Y., Chen, Y., Chen, J., and Ma, P. (2021).
\newblock Best subset selection: Statistical computing meets quantum computing.
\newblock {\em arXiv preprint arXiv:2107.08359}.

\bibitem[Zoufal et~al., 2019]{zoufal2019quantum}
Zoufal, C., Lucchi, A., and Woerner, S. (2019).
\newblock Quantum generative adversarial networks for learning and loading
  random distributions.
\newblock {\em npj Quantum Information}, 5(1):1--9.

\end{thebibliography}
\end{document}